\documentclass[a4paper,USenglish,cleveref, autoref, thm-restate]{lipics-v2021}

\pdfoutput=1 %
\hideLIPIcs  %

\usepackage[utf8]{inputenc}
\usepackage{mathtools}
\usepackage{xparse}
\usepackage[numbers,sort]{natbib}
\usepackage{algorithm}
\usepackage[noend]{algpseudocode}
\algnewcommand\algorithmicinput{\textbf{Input:}}
\algnewcommand\algorithmicoutput{\textbf{Output:}}
\algnewcommand\Input{\item[\algorithmicinput]}
\algnewcommand\Output{\item[\algorithmicoutput]}

\usepackage{tabularx,booktabs,multirow} 
\usepackage{array, makecell}
\usepackage{amsthm}
\usepackage{amsmath}
\usepackage{amssymb}
\usepackage{amsfonts}

\usepackage{tikz}
\usetikzlibrary{arrows,decorations.pathmorphing,decorations.pathreplacing,backgrounds,positioning,fit,matrix,shapes,calc,decorations.text}

\tikzset{
	position/.style args={#1:#2 from #3}{
		at=(#3.#1), anchor=#1+180, shift=(#1:#2)
	},
	vertex/.style = {circle, minimum size=16pt, draw, inner sep=1pt},
	timearc/.style = {draw=black,->,shorten <=1.5pt, shorten >=1.5pt,>=stealth',
		nodes={font=\scriptsize},	
	},
}

\usepackage{subcaption}

\usepackage{doi} %

\newcommand{\NN}{\mathbb{N}} %

\newcommand{\bigO}{\mathcal{O}} %

\newcommand\abs[1]{\left|#1\right|} %

\newcommand{\GG}{\mathcal{G}} %

\newcommand\drp{\textsc{Delay-Robust Route}}
\newcommand\sddrp{\textsc{SD-Delay-Robust Route}}
\newcommand\tddrp{\textsc{TD-Delay-Robust Route}}

\newcommand\mcc{\textsc{Multi-Colored Clique}}
\newcommand\tsat{\textsc{3-SAT}}
\newcommand\mcpsat{\textsc{Multi-Colored Monotone SAT}}
\newcommand\mcpsatshrt{\textsc{MCMSAT}}

\newcommand{\lepoly}{\le_{\textnormal{m}}^{\textnormal{poly}}}

\newcommand{\CNP}{\textnormal{NP}}

\newcommand{\CFPT}{\textnormal{FPT}}
\newcommand{\CWONE}{\textnormal{W[1]}}
\newcommand{\CXP}{\textnormal{XP}}

\newcommand{\under}[1]{G_{\textnormal{u}}(#1)}

\DeclareMathOperator{\tastart}{start}
\DeclareMathOperator{\taend}{end}
\DeclareMathOperator{\tatime}{t}
\DeclareMathOperator{\tatrav}{\lambda}
\DeclareMathOperator{\arrival}{\alpha}

\newcommand{\mvert}{\;\middle\vert\;}

\newcommand{\problemdef}[3]{
		\begin{center}
	\begin{minipage}{0.95\textwidth}
		\noindent
		\textsc{#1}
				\vspace{5pt}\\
				\setlength{\tabcolsep}{3pt}
				\begin{tabularx}{\textwidth}{@{}lX@{}}
						\textbf{Input:} 		& #2 \\
						\textbf{Question:} 	& #3
					\end{tabularx}
	\end{minipage}
		\end{center}
}

\title{Delay-Robust Routes in Temporal Graphs}

\author{Eugen~F\"uchsle}{TU Berlin, Faculty IV, Algorithmics and Computational Complexity, Germany}{fuechsle@campus.tu-berlin.de }{}{}

\author{Hendrik~Molter}{Department of Industrial Engineering and Management, Ben-Gurion~University~of~the~Negev, 
Beer-Sheva, 
Israel}{molterh@post.bgu.ac.il}{https://orcid.org/0000-0002-4590-798X}{Supported by the ISF, grant No.~1070/20.}

\author{Rolf~Niedermeier}{TU Berlin, Faculty IV, Algorithmics and Computational Complexity, Germany}{rolf.niedermeier@tu-berlin.de}{https://orcid.org/0000-0003-1703-1236}{}

\author{Malte~Renken}{TU Berlin, Faculty IV, Algorithmics and Computational Complexity, Germany}{m.renken@tu-berlin.de}{https://orcid.org/0000-0002-1450-1901}{Supported by the DFG, project MATE (NI 369/17).}

\authorrunning{Eugen F\"uchsle, Hendrik Molter, Rolf Niedermeier, and Malte Renken} %

\Copyright{Eugen F\"uchsle, Hendrik Molter, Rolf Niedermeier, and Malte Renken} %



\nolinenumbers %

\begin{document}

\maketitle

\begin{abstract}
Most transportation networks are inherently temporal: Connections (e.g.\ flights, train runs) are only available at certain, scheduled times.
When transporting passengers or commodities, this fact must be considered for the the planning of itineraries.
This has already led to several well-studied algorithmic problems on 
temporal graphs.
The difficulty of the described task is increased by the fact that connections are often unreliable --- in particular, many modes of transportation suffer from occasional delays.
If these delays cause subsequent connections to be missed, the consequences can be severe.
Thus, it is a vital problem to design itineraries that are \emph{robust} to (small) delays.
We initiate the study of this problem from a parameterized complexity perspective
by proving its NP-completeness as well as several hardness and tractability results for natural parameterizations.
\end{abstract}

\section{Introduction}\label{sec:intro}
Finding a path between two vertices in a graph is one of the most fundamental problems in graph algorithmics. In the rise in popularity of \emph{temporal graphs} as a mathematical model~\cite{HotTopicHol15,HotTopicHS19,HotTopicMic16,HotTopicLVM18,HotTopicCas+12}, computing so-called \emph{temporal paths} is one of the most important algorithmic problems in this area.
Herein,
a \text{temporal graph} is a graph whose edges are present only at certain, known points in time.
For our purposes, it is specified by a set~$V$ of vertices and a set~$E$ of time arcs,
where each time arc~$(v,w,t,\lambda) \in E$ consists of a \emph{start vertex}~$v$, an \emph{end vertex}~$w$, a \emph{time label}~$t$, and a traversal time~$\lambda$;
this means that there is a (direct) connection from~$v$ to~$w$ starting at time~$t$ and arriving at time~$t+\lambda$.
Temporal graphs are prime models for many real-world networks:
Social graphs, communication networks, and transportation networks are usually not static but vary over time.

The added dimension of time causes many aspects of connectivity to behave quite differently from static (i.e., non-temporal) graphs.
In particular, the flow of goods or information through a temporal network has to respect time. More formally, it follows a \emph{temporal walk} (or \emph{path}, if every vertex is visited at most once),
i.e., a sequence of time arcs $(v_i, w_i, t_i, \lambda_i)_{i=1}^{\ell}$
where $v_{i+1} = w_i$ and $t_{i+1} \geq t_i + \lambda_i$ for all $i < \ell$.
While inheriting many properties of their static counterparts, temporal walks exhibit certain characteristics that add a new level of complexity to algorithmic problems centered around them. For example, temporal connectivity is not transitive: the existence of a temporal walk from vertex $u$ to $w$ and a temporal walk from $v$ to $w$ does not imply the existence of a temporal walk from $u$ to $w$. Furthermore, the temporal setting allows for several natural notions of an ``optimal'' temporal path~\cite{XuanFJ03}.

As the finding of (optimal) temporal paths and walks constitutes the perhaps most important building block for (algorithmic) analysis of temporal networks,
it has already been studied intensively~\cite{WuCKHHW16EffTempPath,XuanFJ03}.
However, the temporal setting allows to model further natural constraints on temporal walks and paths that do not have a counterpart in the static setting.
For example, recently the study of the computational complexity of finding temporal walks and paths that are subject to some waiting time constraints has been initiated~\cite{TemporalPathsUnderWaitingTime,BentertHNN20TempWalkWaitingTime}.

In this work, we investigate another very natural yet still unstudied temporal path variant, namely so-called \emph{delay-robust} temporal paths.
Real-world networks are often not perfect: Scheduled connections may be canceled or delayed.
This immediately brings up the natural issue of robustness.
To the best of our knowledge, this issue has so far only been analyzed with respect to cancellations \cite{Berman96Vuln},
but not with respect to delays.
We propose a model for delay-robust temporal paths and analyze natural structural and computational problems
occurring in this context.
Our main problem of interest is to determine whether there is a delay-robust temporal path between two vertices in a temporal graph.
\problemdef{\drp{}}
{A temporal graph $\mathcal{G} = (V,E)$, two vertices $s,z \in V$ and $x,\delta \in \mathbb{N}$.}
{Is there an $x$-delay-robust route from~$s$ to~$z$ in~$\GG$?} 

It remains to say how \emph{delay-robustness} is understood.
Although different notions are conceivable, we consider a sequence of vertices (called a \emph{route}) in a temporal graph to be \emph{$x$-delay-robust},
if there is a temporal path visiting the vertices in this sequence even if up to $x$~time arcs are delayed by at most~$\delta$.
We give a formal definition in \cref{sec:prelims}.

This definition is motivated by the fact that changing the vertices may be costly for a number of reasons:
storage or transhipment facilities may need to be newly allocated;
if the new route passes through different jurisdictions, then new authorizations and documents have to be acquired;
insurance policies might not cover alternative routes;
the chosen packaging might no longer be adequate (e.g.\ when switching from rail to air transportation);
or personnel might need to be moved.
All these and many more issues are of much less concern when the chosen route can be kept and only the \emph{schedule} has to be changed.

\subparagraph*{Related Work.} 
Apart from the already mentioned work on finding temporal walks and paths, there has been extensive research on many other connectivity-related problems on temporal graphs~\cite{BussMNR20,Erlebach0K21,KlobasMMNZ21}.
Delays in temporal graphs have been considered as a modification operation to manipulate reachability sets~\cite{DBLP:conf/aaai/DeligkasP20,DBLP:journals/corr/MolTerRenkenZschoche}. 
The individual delay operation considered in the mentioned work delays a single time arc and is similar to our notion of delays.
The deletion of time arcs~\cite{DBLP:journals/corr/MolTerRenkenZschoche,EnrightMMZ19,enright2018deleting}, the deletion of vertices~\cite{staticExp1:journals/jcss/ZschocheFMN20,FluschnikMNRZ20,staticExp2/DBLP:journals/jcss/KempeKK02}, as well as reordering of time arcs~\cite{enright2021assigning} have also been considered as temporal graph modification operations to manipulate the connectivity properties of the temporal graph. The corresponding computational problems in all mentioned work are NP-hard and can be also considered as computing ``robustness measures'' for the connectivity in temporal graphs.

In companion work \cite{ForeseenUnforeseenDelays} we investigate the related problem where we ask whether two vertices remain connected even if up to $x$ time arcs are delayed.
Note that in this setting, the specific temporal path connecting the two vertices can visit different vertices for different delays. We show that this problem can be solved in polynomial time.
We further investigate the problem variant where the delays occur dynamically during the ``journey'' from the start to the destination vertex.
In this case the problem becomes PSPACE-complete if every vertex can be visited at most once and stays polynomial-time solvable, otherwise.

\subparagraph*{Our Contribution.} We introduce the computational problem of finding routes that are robust under delays.
We investigate its computational complexity with a focus on parameterized algorithms and hardness~\cite{DBLP:DowneyFellowsParamCompl,DBLP:CyganEtAlParamAlgs}. 

We first give some structural results in \cref{ch:struct_inv}, including that \textsc{Delay-Robust Path} is solvable in polynomial time if the underlying graph\footnote{The \emph{underlying graph} of a temporal graph is the undirected static graph obtained by connecting all vertices that are connected by a time arc.} is a forest. 
In \cref{subsec:reduction_framework}, we show that \textsc{Delay-Robust Path} is NP-hard even if the underlying graph has constant bandwidth, which implies that it also has constant treewidth. We further show that \textsc{Delay-Robust Path} is W[1]-hard when parameterized by the combination of the feedback vertex number of the underlying graph and the number of delays.
In \cref{sec:algs}, we present our general algorithmic results where we explore how the polynomial-time algorithm for underlying forests can be generalized. We give a polynomial-time algorithm for the case where we have a constant number of delays. We further give two FPT algorithms: one for the underlying feedback edge set number as a parameter and one for the combination of the so-called \emph{timed} feedback vertex number~\cite{TemporalPathsUnderWaitingTime} and the number of delays as a parameter.

\section{Preliminaries}
\label{sec:prelims}

We abbreviate~$\{1, 2, \dots, n\}$ as $[n]$ and $\{n, n+1, \dots, m\}$ as $[n, m]$.
For any time arc~$e = (v, w, t, \lambda_e)$,
we denote the starting and ending vertices as $\tastart(e) = v$ an $\taend(e) = w$, the time label as $\tatime(e) = t$, and the traversal time as $\tatrav(e) = \lambda_e$.
Furthermore, for any vertex~$v$, $\tau_v^+$~denotes the set of time steps where $v$~has outgoing time arcs,
and $\tau_v^-$~denotes the time steps with incoming time arcs.
We set $\tau_v := \tau_v^+ \cup \tau_v^-$.

Given a temporal graph $\GG$, we denote by $T$ the maximum time label of all time arcs in~$\GG$.
When removing all time information and directions from the time arcs of a temporal graph~$\GG = (V, E)$,
the resulting (static \& undirected) graph $\under{\GG} = (V,E')$ with 
$E' = \{\{v,w\} \mid (v,w,t,\lambda) \in E\}$
is called the \emph{underlying graph} of $\GG$.

\subparagraph{Delays.}
We distinguish two different types of delays.
Both are applied to a single time arc~$e$ and delay it by a natural number~$\delta$.
A \emph{starting delay} increases the time label~$\tatime(e)$ by $\delta$
while a \emph{traversal delay} increases the traversal time~$\lambda(e)$ by~$\delta$.
In the example of a railway network, a starting delay would correspond to a delayed departure at a station
whereas a traversal delay would describe a delay occurring on the way between two stations.

For a given set~$D \subseteq E$ of \emph{delayed arcs},
a sequence of time arcs $(v_i, w_i, t_i, \lambda_i)_{i=1}^\ell$ is called
a $D$-starting-delayed temporal walk
resp.\ a $D$-traversal-delayed temporal walk
if it is a temporal walk in the temporal graph obtained from~$\GG$
by applying starting delays resp.\ traversal delays to all time arcs in~$D$.
(We omit $D$ as well as the type of delay when they are clear from context.)
Note that a traversal-delayed temporal walk is always also a temporal walk in~$\GG$,
which is not necessarily true for a starting-delayed temporal walk.

As an example consider the following temporal walk
with edges labeled by $(\tatime(e), \tatrav(e))$:
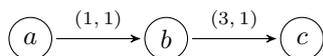
\begin{center}
\begin{tikzpicture}
	\begin{scope}[every node/.style=vertex]
		\node at (0,0) (a) {$a$};
		\node[position=0:8ex from a] (b) {$b$};
		\node[position=0:8ex from b] (c) {$c$};
	\end{scope}
	\draw[timearc]
		(a) edge[edge label={$(1,1)$}] (b)
		(b) edge[edge label={$(3,1)$}] (c);
\end{tikzpicture}
\end{center}
When delaying the first time arc by 1, i.e. when setting $\delta = 1$ and $D = \{(a, b, 1, 1)\}$,
then this is also a starting-delayed as well as a traversal-delayed temporal walk:
Due to the delay, the first time arc arrives in $b$ at time step $2+\delta=3$ which is no later than the departure of the second time arc. 
However, if we instead set~$\delta = 2$ and $D = \{(a, b, 1, 1), (b, c, 3, 1)\}$,
then it is still a starting-delayed temporal walk
but no longer a traversal-delayed temporal walk,
because the first time arc only reaches~$b$ at time~$4$.

We say a sequence~$R$ of vertices forms a \emph{(delayed) route} if there is a (delayed) temporal walk which \emph{follows} $R$,
that is, which visits exactly the vertices of~$R$ in the given order.
Generally, a temporal walk or route from vertex~$s$ to vertex~$z$ is also called a \emph{temporal $(s,z)$-walk} or \emph{$(s,z)$-route}.
A \emph{(delayed) temporal path} is a (delayed) temporal walk where no vertex is visited twice.

\subparagraph{Robustness.}
We say that a temporal route is \emph{traversal-delay-robust} resp.\ \emph{starting-delay-robust} for a given number~$x$ of delays
if it is a $D$-traversal-delayed resp.\ $D$-starting-delayed temporal route for all delay sets~$D$ of size~$\abs{D} \leq x$.
Of course, this also depends on the value of~$\delta$.
An example can be seen in \cref{fig:example_drp}.

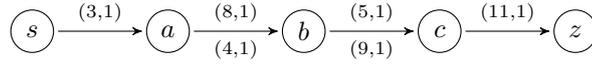
\begin{figure}
	\centering
	\tikzstyle{alter}=[circle, minimum size=16pt, draw, inner sep=1pt] 
	\tikzstyle{majarr}=[draw=black,->,shorten <=1.5pt, shorten >=1.5pt]
	
	\begin{tikzpicture}[auto, >=stealth',shorten <=1pt, shorten >=1pt]
		\node[alter] at (0,0) (s) {$s$};
		\node[alter, position=0:8ex from s] (a) {$a$};
		\node[alter, position=0:8ex from a] (b) {$b$};
		\node[alter, position=0:8ex from b] (c) {$c$};
		\node[alter, position=0:8ex from c] (z) {$z$};		
		
		\draw[majarr] (s) edge node[anchor=south,sloped]{$\scriptstyle (3,1)$} (a);
		\draw[majarr] (a) edge node[midway, anchor=north, sloped]{$\scriptstyle (4,1)$} node[midway, anchor=south, sloped]{$\scriptstyle (8,1)$} (b);
		\draw[majarr] (b) edge node[midway, anchor=south, sloped]{$\scriptstyle (5,1)$} node[midway, anchor=north, sloped]{$\scriptstyle (9,1)$} (c);
		\draw[majarr] (c) edge node[midway, anchor=south, sloped]{$\scriptstyle (11,1)$} (z);		
	\end{tikzpicture}  	
	\caption{An example temporal graph where $(s,a,b,c,z)$ is a traversal- and starting-delay-robust temporal route for $x = 1$ and $\delta = 3$.
		No matter which time arc is delayed, there is a temporal path through these vertices in that order.
		On the other hand, the route ceases to be delay-robust for $\delta \geq 5$,
		as delaying the first arc demonstrates.
	}
	\label{fig:example_drp}
\end{figure}

We now have all the ingredients for the formal definition of our main problem, \drp{}, as given in \cref{sec:intro}.
In this definition, we did not specify whether traversal- or starting-delay is used.
The reason for that is that we will show in \cref{ch:struct_inv} that the distinction is meaningless because both problem variants are equivalent.
In the meantime, however, we will refer to them as \tddrp{} and \sddrp{}.

\section{Structural Results and Recognizing Robust Routes} \label{ch:struct_inv}
In this section, we derive some important properties of delay-robust routes.
\subsection{Structural Results}
We begin by investigating the distinction between walks and paths.
Clearly, from any temporal walk one can obtain a temporal path by eliminating all circular subwalks.
This leads to the following lemma,
which holds for traversal as well as starting delays,
and for all delay sizes~$x$ and delay times~$\delta$
and will come in handy later.
\begin{lemma}
\label{lemma:vwalk_then_vpath}
	Let $s$ and $z$ be two vertices. 
	If there is a delay-robust $(s,z)$-route, then there is a delay-robust $(s,z)$-route without repeated vertices.
\end{lemma}
\begin{proof}
	If there is a delay-robust $(s,z)$-route $R = (v_i)_{i=1}^k$,
	then for each delay of size at most~$x$ there is a delayed temporal walk traversing $v_1, v_2,\ldots, v_k$ in that order.
	Each of these delayed temporal walks can be turned into a delayed temporal walk by eliminating circular subwalks.
	All the delayed temporal paths obtained in this way follow the same sequence of vertices,
	making this sequence a delay-robust $(s,z)$-route without repeated vertices.
\end{proof}

By virtue of \cref{lemma:vwalk_then_vpath},
we will subsequently assume routes to not contain repeated vertices.

Next, we turn towards proving the equivalence of \sddrp{} and \tddrp{}.
We start with some important observations.
The first one is that every traversal-delayed temporal walk is also a starting-delayed temporal walk.
\begin{lemma}
\label{lemma:td_implies_sd}
	Let $P$ be a traversal-delayed temporal walk for some delay set~$D$ and $\delta \in \NN$.
	Then $P$ is also a starting-delayed temporal walk for~$D$ and~$\delta$.
\end{lemma}
\begin{proof}
	Let $P = (v_i, w_i, t_i, \lambda_i)_{i=1}^\ell$.
	This means that
	\[
		t_i + \lambda_i + [e_i \in D] \cdot \delta \le t_{i+1}
	\]
	for all $i \le \ell-1$,
	where $[e_i \in D] = \begin{cases} 1 & \text{if } e_i \in D \\ 0 & \text{otherwise} \end{cases}$ denotes the Iverson bracket.
	Thus
	\[
		t_i + \lambda_i + [e_i \in D] \cdot \delta \le t_{i+1} + [e_{i+1} \in D] \cdot \delta
	\]
	which shows that $P$ is a starting-delayed temporal walk.
\end{proof}

While the converse of \cref{lemma:td_implies_sd} is generally not true, the following weaker statement holds.

\begin{lemma}\label{lemma:minimal_delay_breaks_starting}
	Let $R$ be a route, $\delta \in \NN$ and $D$ a minimal delay set such that $R$~is not a $D$-traversal-delayed route.
	Then $R$ is not a $D$-starting-delayed route either.
\end{lemma}
\begin{proof}
Suppose for contradiction that $R$ was a $D$-starting-delayed route.
Then there is a $D$-starting-delayed temporal walk $P = (e_i)_{i=1}^\ell = (v_i, w_i, t_i, \lambda_i)_{i=1}^\ell$ that follows~$R$, i.e.,
\[
	t_i + \lambda_i + [e_i \in D] \cdot \delta \le t_{i+1} + [e_{i+1} \in D] \cdot \delta
\]
for all $i \leq \ell-1$.
Since $R$ is not a traversal-delayed route,
$P$ is not a traversal-delayed temporal path.
Thus there exists an index $j \leq \ell-1$ with
\[
t_j + \lambda_j + [e_j \in D] \cdot \delta > t_{j+1}
\]
and we may assume~$j$~to be chosen maximally.
This implies that 
\[
		t_{j+1} 
	< 	t_j + \lambda_j + [e_j \in D] \cdot \delta
	\le t_{j+1} + [e_{j+1} \in D] \cdot \delta,
\]
which in turn implies that $e_{j+1} \in D$.
By maximality of~$j$, $P' = (e_i)_{i=j+1}^\ell$ is a traversal-delayed temporal path.
Thus, for any traversal-delayed temporal path $Q = (v_i, w_i, t'_i, \lambda'_i)_{i=1}^j$
following $(v_i)_{i=1}^{j+1}$,
we must have $t'_j + \lambda'_j + [(v_j, w_j, t'_j, \lambda'_j) \in D]\cdot \delta > t_{j+1}$,
for otherwise its concatenation with~$P'$ would contradict the fact that $R$ is not a traversal-delayed route.
Therefore, $R$~is also not a $D'$-traversal-delayed temporal vertex walk,
where $D' = D \setminus \{e_{j+1}\}$.
This contradicts the minimality of~$D$.
\end{proof}

Using \cref{lemma:td_implies_sd,lemma:minimal_delay_breaks_starting}, we can now prove the following.
\begin{theorem}\label{theorem:drp_deltype_equiv}
	\tddrp{} $=$ \sddrp{}.
\end{theorem} 
\begin{proof}
	Let $\mathcal{G} = (V,E)$ be a temporal graph, $s,z \in V$ be a start and an end vertex, and~$x,\delta \in \NN$.
	If $(\GG, s, z, x, \delta)$~is a no-instance of \sddrp{}, then for every $(s,z)$-route~$R$,
	there exists a set~$D$ of $\abs{D} \leq x$~time arcs
	such that there is no $D$-starting-delayed temporal path following~$R$.
	By \cref{lemma:td_implies_sd}, there is then also no $D$-traversal-delayed temporal path following~$R$,
	thus $(\GG, s, z, x, \delta)$ is a no-instance of \tddrp{}.
	
	Conversely, if $(\GG, s, z, x, \delta)$~is a no-instance of \tddrp{},
	then for every $(s,z)$-route~$R$
	there exists a set~$D$ of $\abs{D} \leq x$~time arcs
	such that $R$~is no $D$-traversal-delayed route.
	We may assume~$D$~to be minimal to that respect.
	Then \cref{lemma:minimal_delay_breaks_starting} gives us that $R$~is no $D$-starting-delayed route,
	making $(\GG, s, z, x, \delta)$ a no-instance of \sddrp{}.
\end{proof}

\Cref{theorem:drp_deltype_equiv} allows us now to drop the distinction between the two delay types
and speak simply of \drp{}.
For the remainder of this paper we will mostly work with traversal delays,
for they are slightly easier to handle.

\subsection{Recognizing Robust Routes}
\label{subsec:drp_verification}
Since a route can be followed by an exponential number of different temporal walks,
it is not immediately clear whether delay-robustness can be efficiently checked.
The following theorem says that this is the case,
and \drp{} is thus contained in \CNP{}.

\begin{theorem}\label{thm:polyforest}
	For any given $x, \delta \in \NN$,
	one can determine in $\bigO(nmx^2 + m \log m)$~time whether a given route~$R$ in a temporal graph~$\GG$ is $x$-delay-robust,
	where
	$n$~is the number of vertices of~$R$
	and $m$~is the number of time arcs connecting consecutive vertices of~$R$.
\end{theorem}

\Cref{thm:polyforest} also gives us a polynomial-time algorithm to solve \drp{} on temporal graphs with underlying forest:
As any vertex pair~$(s, z)$ is connected by at most one route,
we only need to test the delay robustness of that route.

In the remainder of this section, we will prove \cref{thm:polyforest}.
The basic idea is that a route is delay-robust for a worst-case delay if and only if it is delay-robust for all delays.
This worst-case delay can be computed in polynomial time using a dynamic program.

First, we introduce the term \emph{earliest arrival time} for a given route.
A route $R = (v_i)_{i=1}^k$ requires that there is at least one temporal path following~$R$.
The \emph{earliest arrival time} is then the arrival time of the temporal path that arrives earliest.
Formally, we define the earliest arrival time as follows.
Let $\mathcal{P} = \{P_i\}_{i=1}^\ell$ be the set of temporal paths following~$R$
with $P_i = (e^{(i)}_1, e^{(i)}_2, \ldots, e^{(i)}_{k-1})$.
The \emph{earliest arrival time} of~$R$ then is defined as the earliest arrival time of any temporal path in~$\mathcal{P}$,
i.e., as
\( \min \left\{ t(e^{(i)}_{k-1}) + \lambda(e^{(i)}_{k-1}) \mvert i \leq \ell \right\} \).
Analogously, if $R = (v_i)_{i=1}^k$ is a delayed route for the delay set~$D\subseteq E$ and delay time~$\delta \in \NN$,
and if $\mathcal{P}$ as above is the set of delayed temporal paths following~$R$,
then the \emph{earliest delayed arrival time} of~$R$ is 
	\(\min \left\{  t(e^{(i)}_{k-1}) + \lambda(e^{(i)}_{k-1}) + [e^{(i)}_{k-1} \in D] \cdot \delta \mvert i \leq \ell \right\} \).

We then define the \emph{worst-case arrival time} of a route~$R = (v_i)_{i=1}^k$
for a given delay size~$x$ and delay time~$\delta$
as the maximum earliest delayed arrival time of~$R$,
taken over all delay sets~$D$ with~$\abs{D} \leq x$.
(If $R$ is not $x$-delay-robust, then we define the worst-case arrival time to be~$\infty$.)

Now that we defined the worst-case arrival time, we show how to compute it.
Let $R_j = (v_i)_{i=1}^j$ denote the prefix routes of~$R$.
The dynamic program computes table entries $A_{R_j}[y]$ iteratively for all $j \leq k$ and $y \leq x$,
where $A_{R_j}[y]$ stores the worst-case arrival time of~$R_j$ for $y$~delays.

We begin with the single-vertex route $R_1 = (v_1)$, setting
\(
	A_{R_1}[y] = 0
\)
for all~$y$
since the empty temporal path is always available to go from~$v_1$ to~$v_1$.
Our goal is then to inductively compute~$A_{R_j}$ from $A_{R_{j-1}}$.

Consider the situation where we want to get from~$v$ to~$w$ in a single step, starting at time~$t$.
Then the set of available time arcs is
\(
	E(v,w,t) = \{ (v,w,t',\lambda) \in E \mid t' \ge t \}
\).
Suppose $E(v, w, t) = \{a_i\}_{i=1}^\ell$ where $\tatime(a_i) + \tatrav(a_i) \leq \tatime(a_{i+1}) + \tatrav(a_{i+1})$ for all~$i$.
Now if up to~$y$~delays occur,
then the latest time at which we will reach~$w$ is
\[
	\arrival(v, w, t, y) := \min\{\tatime(a_1) + \tatrav(a_1) + \delta, \tatime(a_{y+1}) + \tatrav(a_{y+1})\}.
\]
Here, the worst case occurs if $a_1$ through $a_y$ are all delayed.

Using this fact, we can now compute the table entries $A_{R_i}$ from $A_{R_{i-1}}$ as follows.
\begin{equation*}
\label{eq:table-recursion}
	A_{R_{i}}[y] = \max_{0\le y' \le y} \{ \arrival(v_{i-1},v_i,A_{R_{i-1}}[y'],y-y') \}.
\end{equation*}
The idea here is that some number~$y' \leq y$ of delays will occur between $v_{i-1}$ and~$v_i$,
while the other~$y-y'$ delays can occur somewhere along~$R_{i-1}$.

It remains to formally prove that $A_{R}[x]$ contains the solution to the \drp{} instance (\cref{lemma:drw-table-correct})
and that it can be computed in the specified time (\cref{lemma:drp_verif_running_time}).

\begin{lemma}
\label{lemma:drw-table-correct}
	Let $\GG= (V,E)$ be a temporal graph, $R = (v_i)_{i=1}^k$ a route,
	and $j, y, \delta \in \NN$.
	Then $A_{R_j}[y]$ as defined above is the worst-case arrival time of $R_j = (v_i)_{i=1}^j$
	for up to~$y$~delays.
	In particular, $R_j$ is delay-robust for $y$~delays if and only if $A_{R_j}[y] < \infty$.
\end{lemma}
\begin{proof}
	We will prove this by induction. 
	The base case~$i=1$ is clear as mentioned above.

	For the induction step, suppose that the statement holds for~$j-1$.
	Let $t$ be the worst-case arrival time of~$R_j$ and $D$~a delay set of size $\abs{D} \leq y$ for which this delayed arrival time is attained.
	Define $y'$ to be the number of delays in~$D$ that occur along~$R_{j-1}$
	and $t'$~as the delayed arrival time of~$R_{j-1}$ for~$D$.
	
	We may assume $D'$ to be a delay set causing the worst-case arrival time of~$R_{j-1}$ for up to~$y'$~delays,
	because otherwise we could replace~$D'$ with such a set~$\tilde{D}$:
	Clearly, the earliest delayed arrival time of $R_j$ for the delay set $\tilde{D} \cup (D \setminus D')$
	must be at least that for the delay set~$D$.
	
	Due to this, we know by induction hypothesis that
	$t'= A_{R_{j-1}}[y']$.
	Therefore $t = \arrival(v_{j-1}, v_j, t', y - y') = \arrival(v_{j-1}, v_j, A_{R_{j-1}}[y'], y - y') \leq A_{R_j}[y]$
	by definition of $\arrival$.
	But we also have $\arrival(v_{j-1}, v_j, A_{R_{j-1}}[y'], y - y') \geq A_{R_j}[y]$
	or else $D$ would not cause the worst-case arrival time.
	This proves the desired equality.
\end{proof}

The running time of the dynamic program is as follows.
\begin{lemma}
 \label{lemma:drp_verif_running_time}
	Let $\GG$ be a temporal graph, $R = (v_i)_{i=1}^n$ be a route,
	and $x, \delta \in \NN$.
	The dynamic program to compute $A_{R}[x]$ can be executed in
	$\bigO(nmx^2 + m \log m)$~time
	where $m$~is the number of time arcs connecting consecutive vertices of~$R$.
\end{lemma}
\begin{proof}
	We may assume that $\GG=(V,E)$ contains no vertices outside of~$R$,
	thus $n = \abs{V}$ and $m = \abs{E}$.
	First, the time arcs are sorted with respect to the arrival time in $\bigO(m \log m)$ time. 
	
	For the computation of~$a(v,w,t,y)$ we need to filter all time arcs from $v$ to $w$
	that start at time~$t$ or later, obtaining $E(v,w,t)$ (in already sorted order).
	This can be done in $\bigO(m)$~time.
	Afterwards, we can compute $\arrival(v,w,t,y)$ in $\bigO(y)$~time.
	Thus, $\bigO(m)$~time is needed overall to compute $a(v,w,t,y)$.
	
	For the computation of a single table entry $A_{R_i}[y]$,
	the value of $a(v_{i-1},v_i,A_{R_{i-1}}[y'],y-y')$ is computed for all $y' \leq y$.
	This is done for all $i \leq n$ and $y \leq x$.

	Hence, the overall running time is 
	$
	\bigO(n m x^2 + m \log m)
	$.
\end{proof}

Now \cref{thm:polyforest} follows directly from \cref{lemma:drw-table-correct} and \cref{lemma:drp_verif_running_time}.

\section{A Reduction Framework for \drp{}} \label{subsec:reduction_framework}
In this section, we investigate the computational hardness of \drp{} with a particular attention to parameterized hardness with respect to ``distance to forest'' parameters.
The goal is to lay out the ground for potential generalization of the algorithm presented in \cref{subsec:drp_verification}.
We introduce a new problem \mcpsat{} in \cref{sec:mpcsat} and design a polynomial-time reduction to \drp{}.
We will use this as an intermediate problem for reductions from \textsc{3-SAT} and \mcc{} in \cref{subsec:applications} to show \CNP{}-hardness and parameterized hardness results.

\subsection{\mcpsat{}}\label{sec:mpcsat}

The problem \mcpsat{} is a \textsc{Satisfiability} variant where the variables are partitioned into ``color classes'' and only one variable from each color may be set to true. Furthermore, we do not make any assumptions on the Boolean formula other than that all variables appear non-negated. Formally, the we define the problem as follows.

\problemdef{\mcpsat{} (\mcpsatshrt{})}
{Disjoint sets of variables $X_1, X_2, \ldots, X_n$ and a boolean formula $\Phi$ only consisting of positive literals and the operators $\wedge$ and $\vee$.}
{Is there a satisfying truth assignment for $\Phi$ where exactly one variable from each $X_i$ for $i \in [n]$ is true?}

We have the following theorem.

\begin{theorem} \label{cor:poly_red_dj_drp}
	\mcpsatshrt{} $\lepoly$ \drp{}.
\end{theorem}

We describe the reduction behind \cref{cor:poly_red_dj_drp} first, and subsequently prove its correctness in a sequence of lemmas.

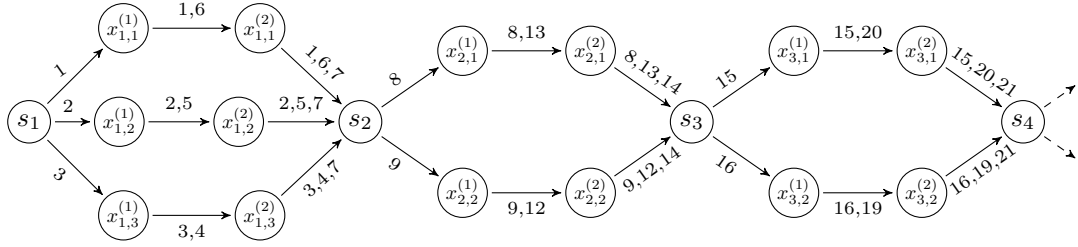
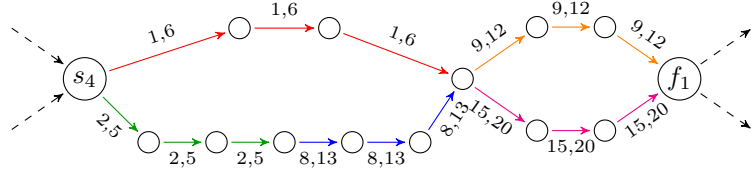
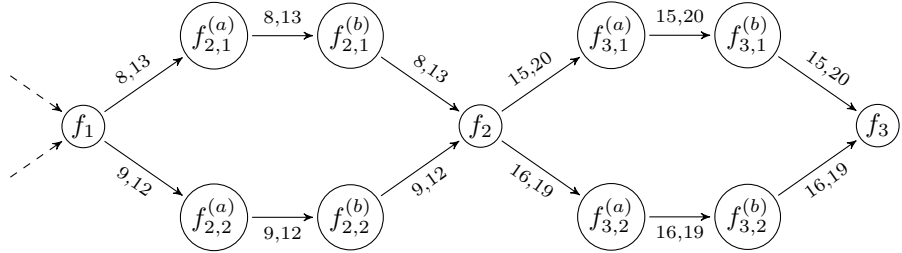
\begin{figure} [t]
	\centering
	\tikzstyle{alter}=[circle, minimum size=16pt, draw, inner sep=1pt] 
	\tikzstyle{majarr}=[draw=black]

	\begin{subfigure}[b]{\textwidth}
		\centering
		
		\begin{tikzpicture}[auto, >=stealth',shorten <=1pt, shorten >=1pt, state/.style={alter, scale=0.75, draw, minimum size=8cm}]
			\tikzstyle{majarr}=[draw=black,->,shorten <=1.5pt, shorten >=1.5pt]
			
			\node[alter] at (0,0) (s1) {$s_1$};
			
			\node[alter, scale=0.75, position=45:10ex from s1] (x1) {$x_{1,1}^{(1)}$};
			\node[alter, scale=0.75, position=0:10ex from x1] (x2) {$x_{1,1}^{(2)}$};
			
			\node[alter, scale=0.75, position=0:5ex from s1] (xxx1) {$x_{1,2}^{(1)}$};
			\node[alter, scale=0.75, position=0:8ex from xxx1] (xxx2) {$x_{1,2}^{(2)}$};

			\node[alter, scale=0.75, position=-45:10ex from s1] (xx1) {$x_{1,3}^{(1)}$};
			\node[alter, scale=0.75, position=0:10ex from xx1] (xx2) {$x_{1,3}^{(2)}$};		
			\node[alter, position=0:25ex from s1] (s2) {$s_2$};

			\draw[majarr] (s1) edge node[pos=0.4, anchor=south, rotate=35]{$\scriptstyle 1$} (x1);
			\draw[majarr] (s1) edge node[pos=0.4, anchor=south]{$\scriptstyle 2$} (xxx1);
			\draw[majarr] (s1) edge node[pos=0.4, anchor=north, rotate=-35]{$\scriptstyle 3$} (xx1);
			
			\draw[majarr] (x1) edge node[anchor=south]{$\scriptstyle 1, 6$} (x2);
			\draw[majarr] (xxx1) edge node[anchor=south]{$\scriptstyle 2, 5$} (xxx2);
			\draw[majarr] (xx1) edge node[anchor=north]{$\scriptstyle 3, 4$} (xx2);
			
			\draw[majarr] (x2) edge node[anchor=south, sloped]{$\scriptstyle 1, 6, 7$} (s2);				
			\draw[majarr] (xxx2) edge node[anchor=south, ]{$\scriptstyle 2, 5, 7$} (s2);
			\draw[majarr] (xx2) edge node[anchor=north, sloped]{$\scriptstyle 3, 4, 7$} (s2);

			\node[alter, scale=0.75, position=35:9ex from s2] (y1) {$x_{2,1}^{(1)}$};
			\node[alter, scale=0.75, position=0:9ex from y1] (y2) {$x_{2,1}^{(2)}$};
			
			\node[alter, scale=0.75, position=-35:9ex from s2] (yy1) {$x_{2,2}^{(1)}$};
			\node[alter, scale=0.75, position=0:9ex from yy1] (yy2) {$x_{2,2}^{(2)}$};		
			\node[alter, position=0:25ex from s2] (s3) {$s_3$};
			
			\draw[majarr] (s2) edge node[pos=0.4, anchor=south, sloped]{$\scriptstyle 8$} (y1);
			\draw[majarr] (s2) edge node[pos=0.4, anchor=north, sloped]{$\scriptstyle 9$} (yy1);
			
			\draw[majarr] (y1) edge node[anchor=south]{$\scriptstyle 8, 13$} (y2);
			\draw[majarr] (yy1) edge node[anchor=north]{$\scriptstyle 9, 12$} (yy2);
			
			\draw[majarr] (y2) edge node[anchor=south, sloped]{$\scriptstyle 8, 13, 14$} (s3);	
			\draw[majarr] (yy2) edge node[anchor=north, sloped]{$\scriptstyle 9, 12, 14$} (s3);

			\node[alter, scale=0.75, position=35:9ex from s3] (z1) {$x_{3,1}^{(1)}$};
			\node[alter, scale=0.75, position=0:9ex from z1] (z2) {$x_{3,1}^{(2)}$};
			
			\node[alter, scale=0.75, position=-35:9ex from s3] (zz1) {$x_{3,2}^{(1)}$};
			\node[alter, scale=0.75, position=0:9ex from zz1] (zz2) {$x_{3,2}^{(2)}$};		
			\node[alter, position=0:25ex from s3] (s4) {$s_4$};
			
			\draw[majarr] (s3) edge node[pos=0.4, anchor=south, sloped]{$\scriptstyle 15$} (z1);
			\draw[majarr] (s3) edge node[pos=0.4, anchor=north, sloped]{$\scriptstyle 16$} (zz1);
			
			\draw[majarr] (z1) edge node[anchor=south]{$\scriptstyle 15, 20$} (z2);
			\draw[majarr] (zz1) edge node[anchor=north]{$\scriptstyle 16, 19$} (zz2);
			
			\draw[majarr] (z2) edge node[anchor=south, sloped]{$\scriptstyle 15, 20, 21$} (s4);	
			\draw[majarr] (zz2) edge node[anchor=north, sloped]{$\scriptstyle 16, 19, 21$} (s4);

			\node[position=35:4ex from s4] (xnext) {};
			\node[position=-35:4ex from s4] (xxnext) {};
			\draw[majarr, dashed] (s4) edge (xnext);
			\draw[majarr, dashed] (s4) edge (xxnext);
		\end{tikzpicture}	
		
		\caption{Selection gadgets for the variable sets with $\abs{X_1} = 3$ and $\abs{X_2} = \abs{X_3} = 2$.}
		\label{fig:mcp_drp_sel_gadget}
	\end{subfigure}
	
	\begin{tikzpicture}
		\node[] at (0,0) (a) {};
	\end{tikzpicture}

	\begin{subfigure}[b]{\textwidth}
		\centering
		
		\begin{tikzpicture}[auto, >=stealth',shorten <=1pt, shorten >=1pt, state/.style={alter, scale=0.75, draw, minimum size=8cm}]
			\tikzstyle{majarr}=[draw=black,->,shorten <=1.5pt, shorten >=1.5pt]
			
			\node[alter] at (0,0) (z0) {$s_4$};
			
			\node[alter, scale=.5, position=-45:10ex from z0] (y11) {};			
			\node[alter, scale=.5, position=0:8ex from y11] (y12) {};
			\node[alter, scale=.5, position=0:8ex from y12] (z11) {};
			\node[alter, scale=.5, position=0:8ex from z11] (z12) {};
			\node[alter, scale=.5, position=0:8ex from z12] (z13) {};

			\node[alter, scale=.5] (x11) at ([shift=({2ex,10ex})]y12) {};
			\node[alter, scale=.5] (x12) at ([shift=({-2ex,10ex})]z12) {};

			\node[alter, scale=.5, position=0:60ex from z0] (i1) {};	

			\node[alter, scale=.5, position=35:12ex from i1] (a11) {};
			\node[alter, scale=.5, position=0:8ex from a11] (a12) {};

			\node[alter, scale=.5, position=-35:12ex from i1] (b11) {};
			\node[alter, scale=.5, position=0:8ex from b11] (b12) {};

			\node[alter, position=0:16ex from i1] (z1) {$f_1$};

			\draw[majarr, color=red] (z0) edge 
				node[anchor=south, color=black, sloped]{$\scriptstyle 1, 6$} (x11);
			\draw[majarr, color=red] (x11) edge 
				node[anchor=south, color=black]{$\scriptstyle 1, 6$} (x12);
			\draw[majarr, color=red] (x12) 
				edge node[anchor=south, color=black, sloped]{$\scriptstyle 1, 6$} (i1);

			\draw[majarr, color=green!60!black] (z0) edge 
				node[anchor=north, color=black, sloped]{$\scriptstyle 2, 5$} (y11);
			\draw[majarr, color=green!60!black] (y11) edge 
				node[anchor=north, color=black]{$\scriptstyle 2, 5$} (y12);
			\draw[majarr, color=green!60!black] (y12) edge 
				node[anchor=north, color=black]{$\scriptstyle 2, 5$} (z11);

			\draw[majarr, color=blue] (z11) edge 
				node[anchor=north, color=black]{$\scriptstyle 8, 13$} (z12);
			\draw[majarr, color=blue] (z12) edge 
				node[anchor=north, color=black]{$\scriptstyle 8, 13$} (z13);
			\draw[majarr, color=blue] (z13) edge 
				node[anchor=north, color=black, sloped]{$\scriptstyle 8, 13$} (i1);

			\draw[majarr, color=orange] (i1) edge 
				node[anchor=south, color=black, sloped]{$\scriptstyle 9, 12$} (a11);
			\draw[majarr, color=orange] (a11) edge 
				node[anchor=south, color=black]{$\scriptstyle 9, 12$} (a12);
			\draw[majarr, color=orange] (a12) edge 
				node[anchor=south, color=black, sloped]{$\scriptstyle 9, 12$} (z1);

			\draw[majarr, color=magenta] (i1) edge 
				node[anchor=north, color=black, sloped]{$\scriptstyle 15, 20$}(b11);
			\draw[majarr, color=magenta] (b11) edge 
				node[anchor=north, color=black]{$\scriptstyle 15, 20$}(b12);
			\draw[majarr, color=magenta] (b12) edge 
				node[anchor=north, color=black, sloped]{$\scriptstyle 15, 20$}(z1);

			\node[position=35:-12ex from z0] (xprev) {};
			\node[position=-35:-12ex from z0] (xxprev) {};
			
			\draw[majarr, dashed] (xprev) edge (z0);
			\draw[majarr, dashed] (xxprev) edge (z0);

			\node[position=35:6ex from z1] (xnext) {};
			\node[position=-35:6ex from z1] (xxnext) {};
			
			\draw[majarr, dashed] (z1) edge (xnext);
			\draw[majarr, dashed] (z1) edge (xxnext);
		\end{tikzpicture}	
		
		\caption{Validation gadgets for $\Phi = ({\color{red}x_{1,1}} \vee ({\color{green!60!black}x_{1,2}} \wedge {\color{blue}x_{2,1}})) \wedge ({\color{orange}x_{2,2}} \vee {\color{magenta}x_{3,1}})$.
		The time arcs belonging to a literal are highlighted in the corresponding color.} 
		\label{fig:mcp_drp_val_gadget}
	\end{subfigure}
	
	\begin{tikzpicture}
		\node[] at (0,0) (a) {};
	\end{tikzpicture}

	\begin{subfigure}[b]{\textwidth}
		\centering
		
		\begin{tikzpicture}[auto, >=stealth',shorten <=1pt, shorten >=1pt, state/.style={alter, scale=0.75, draw, minimum size=8cm}]
			\tikzstyle{majarr}=[draw=black,->,shorten <=1.5pt, shorten >=1.5pt]
			
			\node[alter] at (0,0) (zm) {$f_1$};
			
			\node[alter, position=35:9ex from zm] (x2a) {$f_{2,1}^{(a)}$};		
			\node[alter, position=-35:9ex from zm] (xx2a) {$f_{2,2}^{(a)}$};
			\node[alter, position=0:6ex from x2a] (x2b) {$f_{2,1}^{(b)}$};		
			\node[alter, position=0:6ex from xx2a] (xx2b) {$f_{2,2}^{(b)}$};

			\node[alter, position=-35:9ex from x2b] (f2) {$f_2$};
			
			\node[alter, position=35:9ex from f2] (x3a) {$f_{3,1}^{(a)}$};		
			\node[alter, position=-35:9ex from f2] (xx3a) {$f_{3,2}^{(a)}$};
			\node[alter, position=0:6ex from x3a] (x3b) {$f_{3,1}^{(b)}$};		
			\node[alter, position=0:6ex from xx3a] (xx3b) {$f_{3,2}^{(b)}$};
			
			\node[alter, position=-35:9ex from x3b] (f3) {$f_3$};

			\draw[majarr] (zm) edge node[anchor=south, sloped]{$\scriptstyle 8, 13$} (x2a);
			\draw[majarr] (zm) edge node[anchor=north, sloped]{$\scriptstyle 9, 12$} (xx2a);		
			\draw[majarr] (x2a) edge node[anchor=south]{$\scriptstyle 8, 13$} (x2b);
			\draw[majarr] (xx2a) edge node[anchor=north]{$\scriptstyle 9, 12$} (xx2b);	
			\draw[majarr] (x2b) edge node[anchor=south, sloped]{$\scriptstyle 8, 13$} (f2);		
			\draw[majarr] (xx2b) edge node[anchor=north, sloped]{$\scriptstyle 9, 12$} (f2);
			
			\draw[majarr] (f2) edge node[anchor=south, sloped]{$\scriptstyle 15, 20$} (x3a);
			\draw[majarr] (f2) edge node[anchor=north, sloped]{$\scriptstyle 16, 19$} (xx3a);		
			\draw[majarr] (x3a) edge node[anchor=south]{$\scriptstyle 15, 20$} (x3b);
			\draw[majarr] (xx3a) edge node[anchor=north]{$\scriptstyle 16, 19$} (xx3b);
			\draw[majarr] (x3b) edge node[anchor=south, sloped]{$\scriptstyle 15, 20$} (f3);		
			\draw[majarr] (xx3b) edge node[anchor=north, sloped]{$\scriptstyle 16, 19$} (f3);

			\node[position=35:-12ex from zm] (xprev) {};
			\node[position=-35:-12ex from zm] (xxprev) {};

			\draw[majarr, dashed] (xprev) edge (zm);
			\draw[majarr, dashed] (xxprev) edge (zm);
		\end{tikzpicture}		
		
		\caption{Finalization gadgets for $n=3$.}
		\label{fig:mcp_drp_final_gadget}
	\end{subfigure}
	
	\caption{An example temporal graph resulting from a \mcpsat{} reduction. Dummy time arcs are omitted. The instance has $n=3$ disjoint variable sets.}
	\label{fig:mcp_drp_all_gadgets}
\end{figure}

Let $I = ((X_1, X_2, \ldots, X_n), \Phi)$ be an instance of \mcpsatshrt{}.
We will construct a temporal graph $\GG = (V, E)$ and an instance $I' = (\GG, s,z, \delta, x)$ of \drp{} so that $I$ is a yes-instance of \mcpsatshrt{} if and only if $I'$ is a yes-instance of $\drp{}$. 
All time arcs in $\GG$ have a traversal time of~0.
Thus, we abbreviate time arcs as $3$-tuples $(v,w,t)$. 
In figures we omit the traversal time and label arcs only with their time step.

Let $X_i = \{x_{i,1}, x_{i,2}, \ldots, x_{i, \abs{X_i}}\}$ for all $i\in [n]$. 
Furthermore, let~$m := \max_i \abs{X_i}$ denote the largest cardinality of a variable set $X_i$.
The temporal graph consists of chained \emph{selection gadgets} for each variable set $X_i$, a recursively constructed \emph{validation gadget} and chained \emph{finalization gadgets} for each variable set $X_i$.
The selection gadgets are used to select the variable from $X_i$ that is assigned to true, for each $i \in [n]$.
Then the validation gadgets check whether the formula is satisfied under the selected truth assignment. 
If this is not the case, then a connection breaks at latest in the finalization gadgets and the target vertex can not be reached.
The gadgets use an offset $o_i := (2m+1) \cdot (i-1)$. 
We set the delay time to~$\delta = 1$ and the number of delays to~$ x =2\cdot n-1$. 
\cref{fig:mcp_drp_all_gadgets} shows examples for all gadget types.

\smallskip

\subparagraph*{Selection Gadgets.} 
The selection gadgets are used to select one variable $x_{i,a}$ from each set~$X_i$.
For each set $X_i$, we add a vertex~$s_i$ to~$V$ and one additional vertex~$s_{n+1}$. 
For each set~$X_i$ and each variable $x_{i,a} \in X_i$, we add the vertices $x_{i,a}^{(1)}$ and $x_{i,a}^{(2)}$ to $V$.
Moreover, we add the following time arcs to~$E$ so that there is one route from $s_i$ to $s_{i+1}$ for this variable $x_{i,a}$:
\[
s_i \xrightarrow{o_i+a} x_{i,a}^{(1)} \xrightarrow{o_i+a, o_{i+1}-a} x_{i,a}^{(2)} \xrightarrow{o_i+a, o_{i+1}-a, o_{i+1}} s_{i+1}
\]
Taking this sub-route corresponds to setting the variable $x_{i,a}$ to true.
Additionally, for each of the three underlying arcs we add a \emph{dummy time arc} for each time step $t \in [o_{i+1}-1]$.
If the sub-route $s_i \rightarrow x_{i,a}^{(1)} \rightarrow x_{i,a}^{(2)} \rightarrow s_{i+1}$ is chosen,
then the worst-case arrival time from~$s_1$ to~$s_{i+1}$ is $o_i+a$ for $2\cdot i$ delays and $o_{i+1}-a = o_i + 2\cdot m +1 - a$ for $2\cdot i + 1$ delays.
Any sub-route is a Pareto optimum: while one arrives earlier for $2\cdot i$ delays, another arrives earlier for $2\cdot i + 1$ delays.
We will give a formal proof in \cref{lemma:selection_gadgets_arr_time}.
An example for chained selection
gadgets can be seen in \cref{fig:mcp_drp_sel_gadget}.

\smallskip

\subparagraph*{Validation Gadgets.}
The validation gadgets are used to check whether the formula~$\Phi$ is satisfied under the selected truth assignment. 
We will add a fresh vertex~$f_1$ to~$V$ which is the start of the validation gadgets.
The validation gadget for $\Phi$ will be constructed with~$s_{n+1}$ as a start vertex and $f_1$ as an end vertex.
Given a start vertex $v$ and an end vertex $w$, we can recursively construct the validation gadget for a formula $\Phi$ in the following way: 
\begin{enumerate}
\item
$\Phi = x_{i,a}$ is a single positive literal. 

We add two fresh vertices $\ell_{i,a}^{(1)}$ and $\ell_{i,a}^{(2)}$ to $V$.
We add the following time arcs, so that there is a connection from $v$ to $w$:
\[
	v \xrightarrow{o_i + a, o_{i+1} - a} \ell_{i,a}^{(1)} \xrightarrow{o_i + a, o_{i+1} - a} \ell_{i,a}^{(2)} \xrightarrow{o_i + a, o_{i+1} - a} w
\]
Additionally for all three underlying arcs we add a \emph{dummy time arc} for each time step~$t \in [o_{n+1}-1] \setminus [o_i, o_{i+1}-1]$.
We call this constructed part of the validation gadget a \emph{literal gadget}.
If the variable $x_{i,a}$ has been selected in the selection gadgets, then traversing this literal gadget does not affect the worst-case arrival time with respect to the number of delays.
However, if $x_{i,a}$ has not been selected there is a delay that breaks
the connection at latest in the finalization gadgets.

\item
$\Phi = \Phi_1 \wedge \Phi_2 \wedge \ldots \wedge \Phi_k$ is a conjunction of $k$ sub-formulae.

We add a fresh vertex $c_i$ to $V$ for all $i \in [k-1]$.
Then the validation gadgets for all sub-formulae~$\Phi_i$ are constructed, with $c_{i-1}$ as the start and $c_{i}$ as the end vertex, where~$c_{0} = v$ and $c_{k} = w$.
Thus, the gadgets for the sub-formulae are connected in a row, and to traverse the temporal graph from $v$ to $w$ all gadgets for the sub-formulae have to be traversed.

\item
$\Phi = \Phi_1 \vee \Phi_2 \vee \ldots \vee \Phi_k$ is a disjunction of $k$ sub-formulae.

We construct the validation gadgets for all sub-formulae $\Phi_i$ with $v$ as the start and $w$ as the end vertex.
Thus, the gadgets for the sub-formulae are connected in parallel, and to traverse the temporal graph from $v$ to $w$ one gadget for a sub-formulae has to be traversed.
\end{enumerate}
An example for a valid gadget can be seen in \cref{fig:mcp_drp_val_gadget}.

\smallskip

\subparagraph*{Finalization Gadgets.}
The finalization gadgets are similar to the selection gadgets for all sets $X_2$ to $X_n$.
For each variable set $X_i$ for $i \in [2,n]$ we add a vertex $f_i$ to $V$.
For each variable $x_{i,a} \in X_i$ we add the vertices $f_{i,a}^{(1)}$ and $f_{i,a}^{(2)}$ to $V$ and add the following time arcs:
\[
f_{i-1} \xrightarrow{o_i + a, o_{i+1}-a} f_{i,a}^{(1)} \xrightarrow{o_i + a, o_{i+1}-a} f_{i,a}^{(2)} \xrightarrow{o_i + a, o_{i+1}-a} f_i
\]
Again for all three underlying arcs and each time step $t \in [o_{n+1}-1] \setminus [o_i, o_{i+1}-1]$ we add a dummy time arc.
An example for finalization gadgets can be seen in \cref{fig:mcp_drp_final_gadget}.

The start and end vertices for our \drp{}-instance are $s_1$ and $f_n$, respectively.

It remains to prove that the constructed \drp{} instance is equivalent to the given \mcpsat{} instance.
To this end we will first show that the worst-case arrival time in $s_{n+1}$ (after traversing all selection gadgets) for $2 \cdot (i - 1)$ delays and $2 \cdot (i - 1) + 1$ delays is only affected by the path taken in the $i$-th selection gadget.

\begin{lemma}
 \label{lemma:selection_gadgets_arr_time}
	Let $P$ be a route from $s_1$ to $s_{n+1}$. 
	If the vertices $x_{i,a}^{(1)}$ and $x_{i,a}^{(2)}$ are traversed in $P$, then the worst-case arrival time in $s_{n+1}$ is $o_i + a$ for $2\cdot (i-1)$ delays and~$o_{i+1} - a$ for $2\cdot (i-1) + 1$ delays.
\end{lemma}
\begin{proof}
	We prove this by induction over the prefix-path from $s_1$ to $s_i$.
	
	\proofsubparagraph{Base case ($i = 1$):} 
	Since the path starts in $s_1$ no delays could have occurred. 
	The starting time step is $o_1 = (2m + 1) \cdot(1-1) = 0$.

	\proofsubparagraph{Induction Hypothesis:} 
	A path $P$ from $s_1$ to $s_{i}$ can be reached at time step $o_{i}$ with $2 \cdot (i-1)$ delays and $o_{i} + 1$ with $2 \cdot (i-1) + 1$ delays. 
	For any $j < i$ the arrival times for $2 \cdot (j-1)$ and $2 \cdot (j-1)+1$ delays are $o_j + a$ and $o_{j+1}-a$ respectively depending on which vertices $x_{j,a}^{(1)}$ and $x_{j,a}^{(2)}$ are taken on the path $P$ in the $j$-th variable gadget.
	
	\proofsubparagraph{Induction Step:} 
	First we show, that delaying a dummy time arc is never a strict worst-case. 
	Considering the $i$-th selection gadget there are dummy time arcs for all time steps $t \in [o_{i}-1]$.
	Thus, they are only available if less than $2 \cdot (i-1)$ delays occurred in the previous selection gadgets since this leads to an arrival of at most $o_{j+1}-a \in [o_{i}-1]$ for any $j < i$ and $a \in [m]$. 
	
	\emph{Case 1.} Assume there was a worst-case delay of size $2 \cdot (j-1)$ for some $j<i$.
	The worst-case arrival time in the $i$-th variable gadget is $o_j + a$ for some $a \in [m]$.
	Delaying the dummy time arc at $o_j + a$ will result in an arrival time at $o_j + a + 1$, however, the worst-case arrival time in the $i$-th selection gadget was already $o_{j+1}-a$ for $2 \cdot (j-1)+1$ delays which is at least as late since $o_{j+1}-a = o_j + 2m + 1 - a \ge o_j + a + 1$ for any $a \in [m]$. 
	
	\emph{Case 2.} Assume there was a worst-case delay of size $2 \cdot (j-1) + 1$ for some $j<i$.
	The worst-case arrival time in the $i$-th variable gadget is $o_{j+1} - a$ for some $a \in [m]$.
	Delaying the dummy time arc at $o_{j+1} - a$ will result in an arrival time at $o_{j+1} - a + 1$, however, the worst-case arrival time in the $i$-th selection gadget was already at least $o_{j+1}$ for $2 \cdot (j-1)+1$ delays which is at least as late for any $a \in [m]$.
	
	Thus, in order to have an affect on the arrival time with delays in the i-th variable gadget $2 \cdot (i-1)$ delays need to have occurred yielding an arrival time of $o_{i}$ so that no dummy time arcs are available.
	Note, that although a worst-case delay of size $2 \cdot (i-1) + 1$ leads to an arrival time of $o_i + 1$, the outgoing time arcs at $s_i$ are not before time~$o_i+1$, and thus lead to an arrival at time~$o_i+1$ also for $2 \cdot (i-1)$ delays in later vertices. 
	Thus, we only need to consider a delay of size $2 \cdot (i-1)$.
	The following table shows how the arrival times for a worst-case delay when traversing the vertices that correspond to the variable $x_{i,a}$.
	Hence, the time arcs 
	$$s_i \xrightarrow{o_i+a} x_{i,a}^{(1)} \xrightarrow{o_i+a,o_{i+1}-a} x_{i,a}^{(2)} \xrightarrow{o_i+a,o_{i+1}-a, o_{i+1}} s_{i+1}$$
	are used.
	It can be easily seen that there are no strictly worse ones.
	Delayed arcs are marked in red (delays are already applied) and non delayed time arcs that are taken are marked green.
	\begin{center}
	\begingroup
	\setlength{\tabcolsep}{5pt} 
	\renewcommand{\arraystretch}{1.5} 
		\begin{tabular}{|c|c|c|}
			\hline
			\# Delays & Worst-Case Delay & {\small Arrival Time} \\
			\hline
			{\small$2 \cdot (i-1)$} & $s_i \xrightarrow{{\color{green!50!black}o_i+a}} x_{i,a}^{(1)} \xrightarrow{{\color{green!50!black}o_i+a},o_{i+1}-a} x_{i,a}^{(2)} \xrightarrow{{\color{green!50!black}o_i+a},o_{i+1}-a,o_{i+1}} s_{i+1}$ & $o_i + a$ \\
			\hline
			{\small$2 \cdot (i-1)$ + 1} & $s_i \xrightarrow{{\color{red}o_i+a+1}} x_{i,a}^{(1)} \xrightarrow{o_i+a,{\color{green!50!black}o_{i+1}-a}} x_{i,a}^{(2)} \xrightarrow{o_i+a,{\color{green!50!black}o_{i+1}-a},o_{i+1}} s_{i+1}$ & $o_{i+1} - a$ \\
			\hline
			{\small$2 \cdot i$} & $s_i \xrightarrow{{\color{red}o_i+a+1}} x_{i,a}^{(1)} \xrightarrow{o_i+a,{\color{red}o_{i+1}-a+1}} x_{i,a}^{(2)} \xrightarrow{o_i+a,o_{i+1}-a,{\color{green!50!black}o_{i+1}}} s_{i+1}$ & $o_{i+1}$ \\
			\hline
			{\small$2 \cdot i + 1$} & $s_i \xrightarrow{{\color{red}o_i+a+1}} x_{i,a}^{(1)} \xrightarrow{o_i+a,{\color{red}o_{i+1}-a+1}} x_{i,a}^{(2)} \xrightarrow{o_i+a,o_{i+1}-a,{\color{red}o_{i+1}+1}} s_{i+1}$ & $o_{i+1} + 1$ \\
			\hline		
		\end{tabular}
	\endgroup
	\end{center}
	
	For the last $n$-th variable gadget if $2 \cdot (n-1)$ delays occurred there is only one delay remaining (since the number of allowed delays $x=2\cdot n-1$).
	Thus, if $x_{n,a}^{(1)}$ and $x_{n,a}^{(2)}$ are traversed, an arrival in $z_0$ for $2 \cdot (n-1)$ and $2 \cdot (n-1)+1$ delays is possible at time steps $o_n + a$ and $o_{n+1} - a$ respectively.
\end{proof}

Additionally, we observe that if there is a satisfying assignment for $\Phi$, then there is a path through the validation gadgets that only traverses sub-gadgets corresponding to variables from that satisfying assignment:

\begin{observation} \label{obs:sat_ass_true_lits}
	Let there be a truth assignment where exactly one variable from each $X_i$ for all $i \in [n]$ is set to true.
	This assignment satisfies $\Phi$ if and only if there is a temporal path from $s_n$ to $f_1$ (through the validation gadgets) that only traverses validation sub-gadgets that belong to variables chosen in that assignment.
\end{observation}
This directly follows from the construction of the validation gadgets and the semantics of a boolean formula:
If $\Phi$ is a single literal, then the validation gadget is a single literal gadget. 
If $\Phi$ is a conjunction, then the sub-gadgets for the sub-formulae are connected in a row, so that all need to be traversed. 
If $\Phi$ is a disjunction, then the sub-gadgets for the sub-formulae are connected in parallel, so that one of them needs to be traversed.

Now we will look at the literal gadgets that are traversed after the selection gadgets.
If only literal gadgets are traversed that correspond to variables selected in the selection gadgets, then it will not affect the worst-case arrival time with respect to the number of delays.

\begin{lemma}
 \label{lemma:arrival_time_unchanged_sat_lit}
 	Let $P$ be a route from $s_1$ to any vertex~$z$ of any literal gadget.
 	Suppose for every literal gadget
	\[
	v \xrightarrow{o_i + a, o_{i+1} - a} \ell_{i,a}^{(1)} \xrightarrow{o_i + a, o_{i+1} - a} \ell_{i,a}^{(2)} \xrightarrow{o_i + a, o_{i+1} - a} w
	\]
	traversed by~$P$
	that $P$~also traverses the vertex $x_{i,a}$ corresponding to that literal.
	Then, for any number of delays, the worst-case arrival times of~$P$ at~$z$ and at~$s_{n+1}$ are equal.
\end{lemma}
\begin{proof}
	Assume $v, \ell_{i,a}^{(1)}, \ell_{i,a}^{(2)}, w$ as above are the vertices of the last literal gadget traversed by~$P$.
	Without loss of generality we assume $z = w$.
	By induction on the number of literal gadgets,
	we may assume that $P$ reaches~$v$~and~$s_{n+1}$ at the same time (for any number of delays).
	Thus, due to \cref{lemma:selection_gadgets_arr_time} the arrival time in $v$ is $o_j + a_j$ for $2 \cdot (j-1)$ delays and $o_{j+1} - a_j$ for $2 \cdot (j-1) + 1$ delays for all $j \in [n]$ where $x_{j, a_j}$ is the selected variable from $X_j$, and $a_i = a$ due to our assumption. 
	
	For all delays not of size $2 \cdot (i-1)$ or $2 \cdot (i-1) + 1$ there is a dummy time arc in the literal gadget for the arrival time in $v$, and thus a delay is not a strict worst-case. (see proof of \cref{lemma:selection_gadgets_arr_time}).
	
	For a delay of size $2 \cdot (i-1)$ the arrival time in $v$ is $o_i + a$.
	By applying one additional delay in the literal gadget the arrival in $\ell_{i,a}^{(1)}$ is delayed to $o_i + a + 1$ enforcing to take the time arc at $o_{i+1} - a$ to reach $\ell_{i,a}^{(2)}$.
	However, this is exactly the worst-case arrival time in $v$ with $2 \cdot (i-1) + 1$ delays. 
	Thus, if the additional delay occurred in the selection gadget, then the arrival time would be as late.
	By applying a second additional delay when taking the time arc at $o_{i+1} - a$ to reach $\ell_{i,a}^{(2)}$ one is enforced to take the next dummy time arc $\ell_{i,a}^{(2)} \xrightarrow{o_{i+1}} w$ to reach $w$.
	However, this is even a better arriving time compared to the case that $2 \cdot i$ delays have occurred in the selection gadgets.
	
	For a delay of size $2 \cdot (i-1) + 1$ the arrival time in $v$ is $o_{i+1} - a$.
	By applying one additional delay in the literal gadget the arrival in $\ell_{i,a}^{(1)}$ is delayed to $o_{i+1} - a + 1$ enforcing to take the time arc at $o_{i+1}$ to reach $\ell_{i,a}^{(2)}$.
	However, this is even a better arriving time compared to the case that $2 \cdot i$ have occurred in the selection gadgets. 
\end{proof}

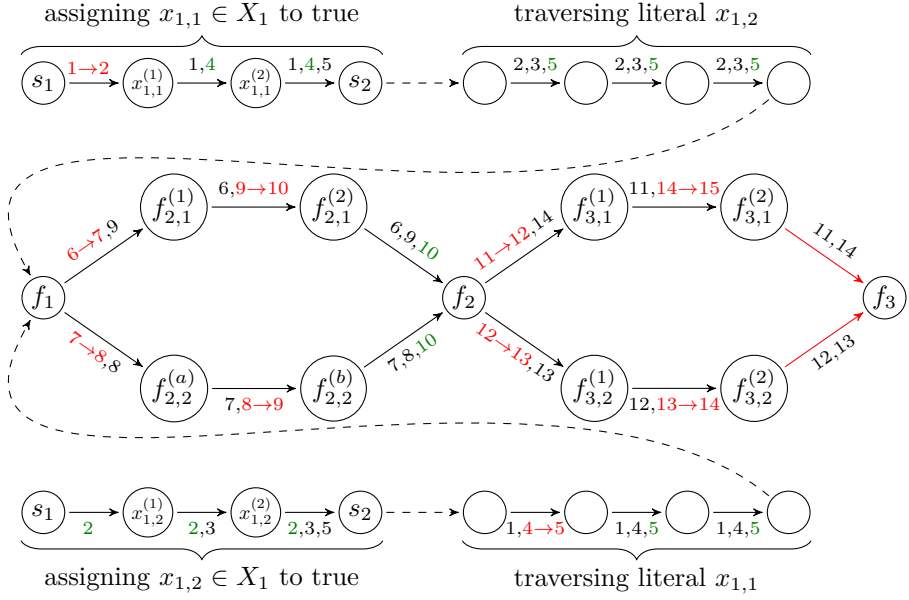
\begin{figure} [t]
	\centering
	\tikzstyle{alter}=[circle, minimum size=16pt, draw, inner sep=1pt] 
	\tikzstyle{majarr}=[draw=black]
	
	\hspace*{-4em}{
	\begin{tikzpicture}[auto, >=stealth',shorten <=1pt, shorten >=1pt, state/.style={alter, draw}, decoration={brace,amplitude=8, raise=4pt}]
		\tikzstyle{majarr}=[draw=black,->,shorten <=1.5pt, shorten >=1.5pt]

		\node[alter] (s1) {$s_1$};
		
		\node[alter, scale=0.75, right = 5ex of s1] (x11) {$x_{1,1}^{(1)}$};
		\node[alter, scale=0.75, right = 5ex of x11] (x12) {$x_{1,1}^{(2)}$};
		\node[alter, right = 5ex of x12] (s2) {$s_2$};

		\node[alter, right = 7ex of s2] (z0) {}; %
		
		\draw[majarr] (s1) edge node[pos=0.4, anchor=south]{$\scriptstyle \color{red}1 \rightarrow 2$} (x11);
		
		\draw[majarr] (x11) edge node[anchor=south]{$\scriptstyle 1, {\color{green!50!black}4}$} (x12);
		
		\draw[majarr] (x12) edge node[anchor=south]{$\scriptstyle 1, {\color{green!50!black}4}, 5$} (s2);	

		\draw [decorate] ([xshift=-3pt]s1.north west) -- node[anchor=south, yshift=0.4cm]{assigning $x_{1,1} \in X_1$ to true} ([xshift=3pt]s2.north east);

		\draw[majarr, dashed] (s2) edge (z0);		
		
		\node[alter, right = 5ex of z0] (c1) {}; %
		\draw[majarr] (z0) edge node[anchor=south]{$\scriptstyle 2,3, {\color{green!50!black}5}$} (c1);

		\node[alter, right = 5ex of c1] (c12) {}; %
		\draw[majarr] (c1) edge node[anchor=south]{$\scriptstyle 2,3, {\color{green!50!black}5}$} (c12);

		\node[alter, right = 5ex of c12] (z1) {}; %
		\draw[majarr] (c12) edge node[anchor=south]{$\scriptstyle 2,3, {\color{green!50!black}5}$} (z1);		

		\node[alter, below = 15ex of s1] (zm) {$f_1$};
		\draw[majarr, dashed] (z1) edge[out=220, in=120] (zm);		
		
		\draw [decorate] ([xshift=-3pt]z0.north west) -- node[anchor=south, yshift=0.4cm]{traversing literal $x_{1,2}$} ([xshift=3pt]z1.north east);
		
		\node[alter, position=35:9ex from zm] (x2a) {$f_{2,1}^{(1)}$};		
		\node[alter, position=0:8ex from x2a] (x2b) {$f_{2,1}^{(2)}$};
		\node[alter, position=-35:9ex from zm] (xx2a) {$f_{2,2}^{(a)}$};
		\node[alter, position=0:8ex from xx2a] (xx2b) {$f_{2,2}^{(b)}$};
		
		\node[alter, position=-35:9ex from x2b] (f2) {$f_2$};
		
		\node[alter, position=35:9ex from f2] (x3a) {$f_{3,1}^{(1)}$};
		\node[alter, position=0:8ex from x3a] (x3b) {$f_{3,1}^{(2)}$};
		\node[alter, position=-35:9ex from f2] (xx3a) {$f_{3,2}^{(1)}$};
		\node[alter, position=-0:8ex from xx3a] (xx3b) {$f_{3,2}^{(2)}$};
		
		\node[alter, position=-35:9ex from x3b] (f3) {$f_3$};
						
		\draw[majarr] (zm) edge node[anchor=south, rotate=35]{$\scriptstyle {\color{red} 6 \rightarrow 7},9$} (x2a);
		\draw[majarr] (zm) edge node[anchor=north, rotate=-35]{$\scriptstyle {{\color{red} 7 \rightarrow 8}, 8}$} (xx2a);		
		\draw[majarr] (x2a) edge node[anchor=south]{$\scriptstyle 6,{\color{red} 9 \rightarrow 10}$} (x2b);		
		\draw[majarr] (x2b) edge node[anchor=south, rotate=-35]{$\scriptstyle 6, 9, {\color{green!50!black} 10}$} (f2);
		\draw[majarr] (xx2a) edge node[anchor=north]{$\scriptstyle 7, {\color{red}8 \rightarrow 9}$} (xx2b);
		\draw[majarr] (xx2b) edge node[anchor=north, rotate=35]{$\scriptstyle 7, 8, {\color{green!50!black}10}$} (f2);
		
		\draw[majarr] (f2) edge node[anchor=south, rotate=35]{$\scriptstyle {\color{red} 11 \rightarrow 12},14$} (x3a);
		\draw[majarr] (f2) edge node[anchor=north, rotate=-35]{$\scriptstyle {{\color{red} 12 \rightarrow 13}}, 13$} (xx3a);
		\draw[majarr] (x3a) edge node[anchor=south]{$\scriptstyle 11,{\color{red} 14 \rightarrow 15}$} (x3b);		
		\draw[majarr, color=red] (x3b) edge node[anchor=south, rotate=-35]{$\scriptstyle \color{black} 11, 14$} (f3);
		\draw[majarr] (xx3a) edge node[anchor=north]{$\scriptstyle 12, {\color{red} 13 \rightarrow 14}$} (xx3b);
		\draw[majarr, color=red] (xx3b) edge node[anchor=north, rotate=35]{$\scriptstyle \color{black} 12, 13$} (f3);

		\node[alter, below=15ex of zm] (bs1) {$s_1$};
		
		\node[alter, scale=0.75, right = 5ex of bs1] (bx11) {$x_{1,2}^{(1)}$};
		\node[alter, scale=0.75, right = 5ex of bx11] (bx12) {$x_{1,2}^{(2)}$};
		\node[alter, right = 5ex of bx12] (bs2) {$s_2$};
		
		\node[alter, right = 7ex of bs2] (bz0) {}; %
		
		\draw[majarr] (bs1) edge node[pos=0.4, anchor=north]{$\scriptstyle \color{green!50!black}2$} (bx11);
		
		\draw[majarr] (bx11) edge node[anchor=north]{$\scriptstyle {\color{green!50!black}2}, 3$} (bx12);
		
		\draw[majarr] (bx12) edge node[anchor=north]{$\scriptstyle {\color{green!50!black}2}, 3, 5$} (bs2);	
		
		\draw [decorate] ([xshift=3pt]bs2.south east) -- node[anchor=north, yshift=-0.4cm]{assigning $x_{1,2} \in X_1$ to true} ([xshift=-3pt]bs1.south west);
		
		\draw[majarr, dashed] (bs2) edge (bz0);		
		
		\node[alter, right = 5ex of bz0] (bc1) {}; %
		\draw[majarr] (bz0) edge node[anchor=north]{$\scriptstyle 1,{\color{red} 4 \rightarrow 5}$} (bc1);
		
		\node[alter, right = 5ex of bc1] (bc12) {}; %
		\draw[majarr] (bc1) edge node[anchor=north]{$\scriptstyle 1,4, {\color{green!50!black}5}$} (bc12);

		\node[alter, right = 5ex of bc12] (bz1) {}; %
		\draw[majarr] (bc12) edge node[anchor=north]{$\scriptstyle 1,4, {\color{green!50!black}5}$} (bz1);		
		
		\draw[majarr, dashed] (bz1) edge[out=140, in=240] (zm);
		
		\draw [decorate] ([xshift=3pt]bz1.south east) -- node[anchor=north, yshift=-0.4cm]{traversing literal $x_{1,1}$} ([xshift=-3pt]bz0.south west);

	\end{tikzpicture}	
	}
	\caption{Example of a connection-breaking delay when a non-satisfied literal sub-path is traversed in the validation gadget. 
	The example instance of \mcpsatshrt{} has three variable sets $X_1, X_2, X_3$, each with at most two variables. (Thus, the offset $o_i = 5 \cdot (i-1)$.)
	The path on the top selects $x_{1,1}$ to true but traverses a sub-path corresponding to the literal $x_{1,2}$, while the path on the bottom selects $x_{1,2}$ to true but traverses a sub-path corresponding to the literal $x_{1,1}$. With only one delay an arrival at time step $5$ can be enforced in $f_1$ and four remaining delays are enough to break the connection in the finalization gadgets. Delays are highlighted in red, non-delayed traversed time arcs are highlighted green. 
	If a connection is broken, then the whole arc is red.}
	\label{fig:unsat_lit_red_demo}
\end{figure}

This is enough to show that the \drp{}-instance and the \mcpsatshrt{}-instance are equivalent:

\begin{lemma} \label{lemma:mcpsat-reduction}
	$((X_1, X_2, \ldots, X_n), \Phi)$ is a yes-instance of \mcpsatshrt{} if and only if the instance $(\GG, s_1,f_n, \delta = 1, x=2\cdot n-1)$ is a yes-instance for \drp{}.
\end{lemma}
\begin{proof}
	~\vspace{-\baselineskip} %
	\proofsubparagraph{($\Rightarrow$):} Let there be a truth assignment that satisfies $\Phi$, where exactly one variable from each $X_i$ for all $i \in [n]$ is set to true. 
	Let $x_{i,a_i}$ be the variable chosen from $X_i$ in that satisfying assignment.
	We choose the route $P_s$ through the selection gadgets from $s_1$ to $s_{n+1}$ that goes through all $x_{i,a_i}^{(1)}$ and $x_{i,a_i}^{(2)}$.
	Now due to \cref{obs:sat_ass_true_lits} there is a temporal path $P_v$ that traverses the validation gadgets from $s_{n+1}$ to $f_1$ that only traverses literal gadgets of variables selected in the selection gadgets.
	Hence, $P_s \circ P_v$ is delay-robust and the arrival times with respect to the number of delays in $s_{n+1}$ are the same as in $f_{1}$ due to \cref{lemma:selection_gadgets_arr_time} and \cref{lemma:arrival_time_unchanged_sat_lit}.
	Now for the finalization gadgets we choose the route~$P_f$ going through the vertices $f_{i,a_i}^{(1)}$ and $f_{i,a_i}^{(2)}$ for all $i \in [2,n]$.
	It can be seen that the sub-paths in the finalization gadgets
	\[
		f_{i-1} \xrightarrow{o_i + a_i, o_{i+1}-a_i} f_{i,a_i}^{(1)} \xrightarrow{o_i + a_i, o_{i+1}-a_i} f_{i,a_i}^{(2)} \xrightarrow{o_i + a_i, o_{i+1}-a_i} f_i
	\]
	are similar to the satisfied literal gadgets.
	Thus, the traversal of the finalization gadget does not worsen the arrival time with respect to the number of delays.
	Therefore, $P_s \circ P_v \circ P_f$ is a delay-robust route from $s_1$ to $f_n$ for any delay of size at most $2 \cdot n - 1$ and $\delta = 1$.
	
	\proofsubparagraph{($\Leftarrow$):} Assume there does not exist a truth assignment where exactly one variable from each $X_i$ for all $i \in [n]$ is set to true that satisfies $\Phi$. 
	Thus, due to \cref{obs:sat_ass_true_lits} for any assignment all paths through the validation gadgets traverse at least one literal gadget corresponding to an unsatisfied literal.
	Therefore, for any path $P$ from $s_1$ to $f_1$ traversing the selection and validation gadgets there is an $i \in [n]$ so that $P$ traverses a literal gadget 
	\[
	v \xrightarrow{o_i + a, o_{i+1} - a} \ell_{i,a}^{(1)} \xrightarrow{o_i + a, o_{i+1} - a} \ell_{i,a}^{(2)} \xrightarrow{o_i + a, o_{i+1} - a} w
	\]
	while in the $i$-th selection gadget the vertices $s_i \rightarrow x_{i,b}^{(1)} \rightarrow x_{i,b}^{(2)} \rightarrow s_{i+1}$ are traversed with $a \ne b$.
	Without loss of generality let this be the first unsatisfied literal gadget.
	By doing a case distinction we can show, that with only $2 \cdot (i-1) + 1$ delays an arrival of at least $o_{i+1}$ in $w$ can be enforced.
	
	If $a < b$, then
	due to \cref{lemma:selection_gadgets_arr_time} and \cref{lemma:arrival_time_unchanged_sat_lit}, there is a delay of size $2 \cdot (i-1)$ yielding an arrival time of $o_i + b$ in $v$.
	Since $a < b$ implies $o_i + a < o_i + b$ the time arc $v \xrightarrow{o_i + a} \ell_{i,a}^{(1)}$ cannot be taken.
	The next possible time arc is $v \xrightarrow{o_{i+1} - a} \ell_{i,a}^{(1)}$.
	By delaying that arc the arrival in $\ell_{i,a}^{(1)}$ is at $o_{i+1} - a$ and to reach the next vertices $\ell_{i,a}^{(2)}$ and $w$ the dummy time arc at $o_{i+1}$ needs to be taken. 
	Note that the total number of delays used is $2 \cdot (i-1) + 1$.
	
	If $a > b$, then
	due to \cref{lemma:selection_gadgets_arr_time} and \cref{lemma:arrival_time_unchanged_sat_lit}, there is a delay of size $2 \cdot (i-1) + 1$ yielding an arrival time of $o_{i+1} - b$ in $v$.
	Since $a > b$ implies $o_{i+1} - a < o_{i+1} - b$ none of the time arcs $v \xrightarrow{o_i + a, o_{i+1} - a} \ell_{i,a}^{(1)}$ can be taken.
	The next possible time arc is the dummy time arc at $o_{i+1}$ which leads to an arrival time of $o_{i+1}$ in $w$ using $v \xrightarrow{o_{i+1}} \ell_{i,a}^{(1)} \xrightarrow{o_{i+1}} \ell_{i,a}^{(2)} \xrightarrow{o_{i+1}} w$.
	
	Note that there are no dummy time arcs for $o_{n+1}$.
	Thus, the connection breaks in the literal gadget if $i = n$.
	Since the arrival time in $w$ is $o_{i+1}$ for $2 \cdot (i-1) + 1$ delays, the arrival time in $f_{i}$ is at least $o_{i+1}$ for $2 \cdot (i-1) + 1$ delays.
	Now by applying 2 additional delays in the $(i+1)$-st finalization gadget
	\[
		f_{i} \xrightarrow{{\color{red}o_{i+1} + a \rightarrow o_{i+1} + a+1}} f_{i+1,a}^{(1)} \xrightarrow{{\color{gray}o_{i+1} + a}, {\color{red}o_{i+2}-a \rightarrow o_{i+2}-a+1}} f_{i+1,a}^{(2)} \xrightarrow{{\color{gray}o_{i+1} + a, o_{i+2}-a}, o_{i+2}} f_{i+1}
	\]
	the arrival time in $f_{i+1}$ is $o_{i+2}$ with a total of $2 \cdot i + 1$ delays. (Red labels indicate a delay, gray labels cannot be taken due to a previous delay.)
	This can be repeated until the last finalization gadget where one arrives earliest at $o_n$ with $2\cdot(n-2)+1$ delays used, and thus two delays left. 
	Similar as before two delays can be applied, but for the $n$-th finalization gadget there is no dummy time arc of time step $o_{n+1}$, and thus the connection breaks with a total of $2\cdot (n-1) + 1 = 2 \cdot n - 1$ delays.
	Hence, there is no delay-robust path from $s_1$ to $f_n$ with $2 \cdot n - 1$ delays and $\delta = 1$, and thus the instance is a no-instance.
\end{proof}

An example how the delay breaks when an unsatisfied literal-gadget is traversed can be seen in \cref{fig:unsat_lit_red_demo}.
We now show that the reduction can be performed in polynomial time.

\begin{lemma}
\label{lem:redptime}
	Given an instance $((X_1, X_2, \ldots, X_n), \Phi)$ of \mcpsatshrt{} with $m = \max_i \abs{X_i}$ and $\ell$ the number of literals in $\Phi$,
	the temporal graph $\GG$ can be constructed in polynomial time with respect to $n, m$ and $\ell$.
\end{lemma}
\begin{proof}
	The graph $\GG$ contains the $\bigO(n)$ vertices $s_1, s_2, \ldots, s_n$ and $f_2, f_3, \ldots, f_n$. 
	For each $i \in [n]$ and $x_{i,a} \in X_i$ there are two additional vertices in the selection and finalization gadgets each resulting in $\bigO(n \cdot m)$ vertices.
	Furthermore, each variable $x_{i,a} \in X_i$ for each $i \in [n]$ introduces at most $3 \cdot o_{n+1} \in \bigO(n \cdot m)$ time arcs giving a total of $\bigO(n^2 \cdot m^2)$ time arcs.
	Each literal in $\Phi$ introduces at most four new vertices in $\GG$ and at most $3 \cdot o_{n+1} \in \bigO(n \cdot m)$ time arcs.
	Thus, the validation gadget introduces $\bigO(\ell)$ vertices and $\bigO(\ell \cdot n \cdot m)$ time arcs.
	In total, $\GG$ has $\bigO(\ell + n \cdot m)$ vertices and $\bigO(\ell \cdot n \cdot m + n^2 \cdot m^2)$ time arcs. 
	The temporal graph can be constructed by iterating over the variables and clauses once.
\end{proof}

Hence, we have constructed a valid polynomial-time reduction and \cref{cor:poly_red_dj_drp} follows from \cref{lemma:mcpsat-reduction,lem:redptime}.

\subsection{Applications of the Framework}\label{subsec:applications}

Next, we use our previous result that \mcpsatshrt{}~$\lepoly$~\drp{} (\cref{cor:poly_red_dj_drp}) to show that \drp{} is \CNP{}-complete even if the underlying graph has bandwidth 3.
The \emph{bandwidth} $\textnormal{bw}(G)$ of a graph $G$ is the smallest number $b$ such that the vertices of $G$ can be placed at distinct integer points along a line
so that the length of the longest edge is $b$.
The bandwidth of a graph upper-bounds both the graph's pathwidth and treewidth~\cite{sorgeWellerGraphHierarchy}.
Formally, we show the following result by using an appropriate polynomial-time reduction from the \CNP{}-complete \tsat{} problem~\cite{Karp72NPComplete} to \mcpsatshrt{}.

\begin{theorem} \label{theorem:drp_paranp}
	\drp{} is NP-complete for all fixed $\delta\ge 1$, maximum traversal times $\lambda_{\max}\ge 0$, and bandwidths of the underlying graph $\textnormal{bw}(\under{\GG})\ge 3$ .
\end{theorem}

Let $\Phi$ be an instance of \tsat{} with the variables $x_1, x_2, \ldots, x_n$.
We will construct an instance $((X_1, X_2, \ldots, X_n), \Phi')$ of \mcpsatshrt{} that is a yes-instance if and only if $\Phi$ is a yes-instance of \tsat{}.
The construction is straightforward.
Let $X_i = \{x_i, \bar{x}_i\}$ for all $i \in [n]$.
Note, that in this case $\bar{x}_i$ is a variable and not a negative literal.
Furthermore, $\Phi' = \Phi$, but any negative literal $\bar{x}_i$ from $\Phi$ corresponds to the variable $\bar{x}_i$, and thus is not negated in $\Phi'$.
The formula $\Phi'$ only contains the operators $\wedge$ and $\vee$ since $\Phi$ is in conjunctive normal form.
Hence, $\Phi'$ is a valid formula for the problem \mcpsatshrt{}.

\begin{lemma} \label{lemma:3sat_to_djpsat}
	$\Phi$ is a yes-instance of \tsat{} if and only if $((X_1, X_2, \ldots, X_n), \Phi')$ is a yes-instance of \mcpsatshrt.
\end{lemma}
\begin{proof}
	~\vspace{-\baselineskip} %
	\proofsubparagraph{($\Rightarrow$):} Assume $\Phi$ is a yes-instance of \tsat{}.
	Hence, there is a truth assignment
$\hat{x}_1, \hat{x}_2, \ldots, \hat{x}_n$ with $\hat{x}_i \in \{x_i, \bar{x}_i \}$ that satisfies $\Phi$.
	For \mcpsatshrt{} exactly one variable from each set $X_i$ with $i \in n$ has to be set to true.
	If $\hat{x}_i = x_i$, then we set $x_i \in X_i$ to true, and thus $\bar{x}_i \in X_i$ is false. 
	If $\hat{x}_i = \bar{x}_i$, then we set $\bar{x}_i \in X_i$ to true, and thus $x_i \in X_i$ is false. 
	Hence, for any literal in $\Phi$ the corresponding positive literal in $\Phi'$ is evaluated to the same truth-value.
	This also implies that $\Phi$ and $\Phi'$ are evaluated to the same truth-value since the literals are linked by conjunctions and disjunctions in the same way.
	Since $\Phi$ is satisfied, the variable selection from the sets $X_i$ with $i \in [n]$ also satisfies $\Phi'$.

	\proofsubparagraph{($\Leftarrow$):} Assume $((X_1, X_2, \ldots, X_n), \Phi')$ is a yes-instance of \mcpsatshrt.
	Thus, there is a truth assignment that sets exactly one variable from each set $X_i$ with $i \in [n]$ to true.
	Let $\hat{x}_i$ be the selected variable from $X_i$.
	We will also use this assignment for our \tsat{}-instance: If $\hat{x}_i = x_i$, then we assign $x_i$ to true, otherwise to false.
	Hence, a positive literal in $\Phi$ is evaluated to true if and only if the corresponding literal in $\Phi$ is evaluated to true.
	This also implies that $\Phi'$ and $\Phi$ are evaluated to the same truth-value since the literals are linked by conjunctions and disjunctions in the same way.
	Since $\Phi'$ is satisfied under the selected variables the truth assignment also satisfies $\Phi$.
\end{proof}

\begin{figure} [t]
	\centering
	\tikzstyle{alter}=[circle, minimum size=16pt, draw, inner sep=1pt] 
	\tikzstyle{majarr}=[draw=black]

	\begin{subfigure}[b]{\textwidth}
		\centering
		
		\begin{tikzpicture}[auto, >=stealth',shorten <=1pt, shorten >=1pt, state/.style={alter, scale=0.75, draw, minimum size=8cm}]
			\tikzstyle{majarr}=[draw=black,->,shorten <=1.5pt, shorten >=1.5pt]
			
			\node[alter, fill={rgb,255:red,150; green,255; blue,150}, scale=0.25] at (0.0, 0) (sel1s) {};
\node[alter, fill=black, scale=0.25] at (2.7, 0) (sel1z) {};
\node[alter, fill=black, scale=0.25] at (0.9, 0.5) (sel101) {};
\node[alter, fill=black, scale=0.25] at (1.8, 0.5) (sel102) {};
\draw[majarr] (sel1s) edge (sel101);
\draw[majarr] (sel101) edge (sel102);
\draw[majarr] (sel102) edge (sel1z);
\node[alter, fill=black, scale=0.25] at (0.9, -0.5) (sel111) {};
\node[alter, fill=black, scale=0.25] at (1.8, -0.5) (sel112) {};
\draw[majarr] (sel1s) edge (sel111);
\draw[majarr] (sel111) edge (sel112);
\draw[majarr] (sel112) edge (sel1z);
\node[alter, fill=black, scale=0.25] at (5.4, 0) (sel2z) {};
\node[alter, fill=black, scale=0.25] at (3.6, 0.5) (sel201) {};
\node[alter, fill=black, scale=0.25] at (4.5, 0.5) (sel202) {};
\draw[majarr] (sel1z) edge (sel201);
\draw[majarr] (sel201) edge (sel202);
\draw[majarr] (sel202) edge (sel2z);
\node[alter, fill=black, scale=0.25] at (3.6, -0.5) (sel211) {};
\node[alter, fill=black, scale=0.25] at (4.5, -0.5) (sel212) {};
\draw[majarr] (sel1z) edge (sel211);
\draw[majarr] (sel211) edge (sel212);
\draw[majarr] (sel212) edge (sel2z);
\node[alter, fill=black, scale=0.25] at (7.2, 0) (sel3s) {};
\node[alter, fill=black, scale=0.25] at (9.9, 0) (sel3z) {};
\node[alter, fill=black, scale=0.25] at (8.1, 0.5) (sel301) {};
\node[alter, fill=black, scale=0.25] at (9.0, 0.5) (sel302) {};
\draw[majarr] (sel3s) edge (sel301);
\draw[majarr] (sel301) edge (sel302);
\draw[majarr] (sel302) edge (sel3z);
\node[alter, fill=black, scale=0.25] at (8.1, -0.5) (sel311) {};
\node[alter, fill=black, scale=0.25] at (9.0, -0.5) (sel312) {};
\draw[majarr] (sel3s) edge (sel311);
\draw[majarr] (sel311) edge (sel312);
\draw[majarr] (sel312) edge (sel3z);
\path (sel2z) -- node[auto=false]{\textbf{\ldots}} (sel3s);
\node[scale=0.25] at (10.8, 0.5) (sel3zd0) {};
\draw[majarr, dashed] (sel3z) edge (sel3zd0);
\node[scale=0.25] at (10.8, 0.0) (sel3zd1) {};
\draw[majarr, dashed] (sel3z) edge (sel3zd1);
\node[scale=0.25] at (10.8, -0.5) (sel3zd2) {};
\draw[majarr, dashed] (sel3z) edge (sel3zd2);

		\end{tikzpicture}	
		
		\caption{Selection gadgets: $n$ chained selection gadgets to select the assignment for each variable $x_i$ for $i\in [n]$.}		
	\end{subfigure}
	
	\begin{tikzpicture}
		\node[] at (0,0) (a) {};
	\end{tikzpicture}

	\begin{subfigure}[b]{\textwidth}
		\centering
		
		\begin{tikzpicture}[auto, >=stealth',shorten <=1pt, shorten >=1pt, state/.style={alter, scale=0.75, draw, minimum size=8cm}]
			\tikzstyle{majarr}=[draw=black,->,shorten <=1.5pt, shorten >=1.5pt]
			
\node[alter, fill=black, scale=0.25] at (0.0, 0) (val1s) {};
\node[alter, fill=black, scale=0.25] at (2.7, 0) (val1z) {};
\node[alter, fill=black, scale=0.25] at (0.9, 0.75) (val101) {};
\node[alter, fill=black, scale=0.25] at (1.8, 0.75) (val102) {};
\draw[majarr] (val1s) edge (val101);
\draw[majarr] (val101) edge (val102);
\draw[majarr] (val102) edge (val1z);
\node[alter, fill=black, scale=0.25] at (0.9, 0.0) (val111) {};
\node[alter, fill=black, scale=0.25] at (1.8, 0.0) (val112) {};
\draw[majarr] (val1s) edge (val111);
\draw[majarr] (val111) edge (val112);
\draw[majarr] (val112) edge (val1z);
\node[alter, fill=black, scale=0.25] at (0.9, -0.75) (val121) {};
\node[alter, fill=black, scale=0.25] at (1.8, -0.75) (val122) {};
\draw[majarr] (val1s) edge (val121);
\draw[majarr] (val121) edge (val122);
\draw[majarr] (val122) edge (val1z);
\node[alter, fill=black, scale=0.25] at (5.4, 0) (val2z) {};
\node[alter, fill=black, scale=0.25] at (3.6, 0.75) (val201) {};
\node[alter, fill=black, scale=0.25] at (4.5, 0.75) (val202) {};
\draw[majarr] (val1z) edge (val201);
\draw[majarr] (val201) edge (val202);
\draw[majarr] (val202) edge (val2z);
\node[alter, fill=black, scale=0.25] at (3.6, 0.0) (val211) {};
\node[alter, fill=black, scale=0.25] at (4.5, 0.0) (val212) {};
\draw[majarr] (val1z) edge (val211);
\draw[majarr] (val211) edge (val212);
\draw[majarr] (val212) edge (val2z);
\node[alter, fill=black, scale=0.25] at (3.6, -0.75) (val221) {};
\node[alter, fill=black, scale=0.25] at (4.5, -0.75) (val222) {};
\draw[majarr] (val1z) edge (val221);
\draw[majarr] (val221) edge (val222);
\draw[majarr] (val222) edge (val2z);
\node[alter, fill=black, scale=0.25] at (7.2, 0) (val3s) {};
\node[alter, fill=black, scale=0.25] at (9.9, 0) (val3z) {};
\node[alter, fill=black, scale=0.25] at (8.1, 0.75) (val301) {};
\node[alter, fill=black, scale=0.25] at (9.0, 0.75) (val302) {};
\draw[majarr] (val3s) edge (val301);
\draw[majarr] (val301) edge (val302);
\draw[majarr] (val302) edge (val3z);
\node[alter, fill=black, scale=0.25] at (8.1, 0.0) (val311) {};
\node[alter, fill=black, scale=0.25] at (9.0, 0.0) (val312) {};
\draw[majarr] (val3s) edge (val311);
\draw[majarr] (val311) edge (val312);
\draw[majarr] (val312) edge (val3z);
\node[alter, fill=black, scale=0.25] at (8.1, -0.75) (val321) {};
\node[alter, fill=black, scale=0.25] at (9.0, -0.75) (val322) {};
\draw[majarr] (val3s) edge (val321);
\draw[majarr] (val321) edge (val322);
\draw[majarr] (val322) edge (val3z);
\path (val2z) -- node[auto=false]{\textbf{\ldots}} (val3s);
\node[scale=0.25] at (-0.9, 0.75) (val1sd0) {};
\draw[majarr, dashed] (val1sd0) edge (val1s);
\node[scale=0.25] at (-0.9, 0.0) (val1sd1) {};
\draw[majarr, dashed] (val1sd1) edge (val1s);
\node[scale=0.25] at (-0.9, -0.75) (val1sd2) {};
\draw[majarr, dashed] (val1sd2) edge (val1s);
\node[scale=0.25] at (10.8, 0.75) (val3zd0) {};
\draw[majarr, dashed] (val3z) edge (val3zd0);
\node[scale=0.25] at (10.8, 0.0) (val3zd1) {};
\draw[majarr, dashed] (val3z) edge (val3zd1);
\node[scale=0.25] at (10.8, -0.75) (val3zd2) {};
\draw[majarr, dashed] (val3z) edge (val3zd2);

		\end{tikzpicture}	
		
		\caption{Validation gadget: $m$ chained sub-gadgets to choose a literal for each clause.}		
	\end{subfigure}
	
	\begin{tikzpicture}
		\node[] at (0,0) (a) {};
	\end{tikzpicture}

	\begin{subfigure}[b]{\textwidth}
		\centering
		
		\begin{tikzpicture}[auto, >=stealth',shorten <=1pt, shorten >=1pt, state/.style={alter, scale=0.75, draw, minimum size=8cm}]
			\tikzstyle{majarr}=[draw=black,->,shorten <=1.5pt, shorten >=1.5pt]
			
\node[alter, fill=black, scale=0.25] at (0.0, 0) (fin1s) {};
\node[alter, fill=black, scale=0.25] at (2.7, 0) (fin1z) {};
\node[alter, fill=black, scale=0.25] at (0.9, 0.5) (fin101) {};
\node[alter, fill=black, scale=0.25] at (1.8, 0.5) (fin102) {};
\draw[majarr] (fin1s) edge (fin101);
\draw[majarr] (fin101) edge (fin102);
\draw[majarr] (fin102) edge (fin1z);
\node[alter, fill=black, scale=0.25] at (0.9, -0.5) (fin111) {};
\node[alter, fill=black, scale=0.25] at (1.8, -0.5) (fin112) {};
\draw[majarr] (fin1s) edge (fin111);
\draw[majarr] (fin111) edge (fin112);
\draw[majarr] (fin112) edge (fin1z);
\node[alter, fill=black, scale=0.25] at (5.4, 0) (fin2z) {};
\node[alter, fill=black, scale=0.25] at (3.6, 0.5) (fin201) {};
\node[alter, fill=black, scale=0.25] at (4.5, 0.5) (fin202) {};
\draw[majarr] (fin1z) edge (fin201);
\draw[majarr] (fin201) edge (fin202);
\draw[majarr] (fin202) edge (fin2z);
\node[alter, fill=black, scale=0.25] at (3.6, -0.5) (fin211) {};
\node[alter, fill=black, scale=0.25] at (4.5, -0.5) (fin212) {};
\draw[majarr] (fin1z) edge (fin211);
\draw[majarr] (fin211) edge (fin212);
\draw[majarr] (fin212) edge (fin2z);
\node[alter, fill=black, scale=0.25] at (7.2, 0) (fin3s) {};
\node[alter, fill=blue, scale=0.25] at (9.9, 0) (fin3z) {};
\node[alter, fill=black, scale=0.25] at (8.1, 0.5) (fin301) {};
\node[alter, fill=black, scale=0.25] at (9.0, 0.5) (fin302) {};
\draw[majarr] (fin3s) edge (fin301);
\draw[majarr] (fin301) edge (fin302);
\draw[majarr] (fin302) edge (fin3z);
\node[alter, fill=black, scale=0.25] at (8.1, -0.5) (fin311) {};
\node[alter, fill=black, scale=0.25] at (9.0, -0.5) (fin312) {};
\draw[majarr] (fin3s) edge (fin311);
\draw[majarr] (fin311) edge (fin312);
\draw[majarr] (fin312) edge (fin3z);
\path (fin2z) -- node[auto=false]{\textbf{\ldots}} (fin3s);
\node[scale=0.25] at (-0.9, 0.5) (fin1sd0) {};
\draw[majarr, dashed] (fin1sd0) edge (fin1s);
\node[scale=0.25] at (-0.9, 0.0) (fin1sd1) {};
\draw[majarr, dashed] (fin1sd1) edge (fin1s);
\node[scale=0.25] at (-0.9, -0.5) (fin1sd2) {};
\draw[majarr, dashed] (fin1sd2) edge (fin1s);

		\end{tikzpicture}		
		
		\caption{Finalization gadgets: $n-1$ chained finalization gadgets.}
	\end{subfigure}
	
	\caption{The temporal graph resulting from the \tsat{} $\lepoly$ \drp{} reduction. 
	The start vertex is highlighted light green, the end vertex is highlighted dark blue. 
	All time labels are ommited. The number of variables of the \tsat{}-instance is $n$, the number of clauses is $m$.}
	\label{fig:3sat_drp_red_graph}
\end{figure}
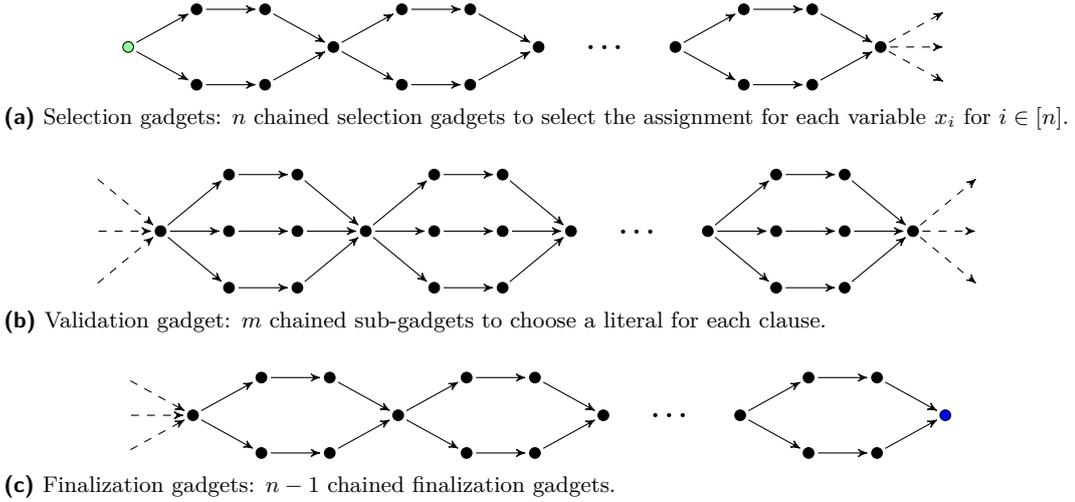

\begin{proof}[Proof of \cref{theorem:drp_paranp}]
We now obtain \cref{theorem:drp_paranp} from \cref{lemma:3sat_to_djpsat}, observing that the reduction can clearly be performed in polynomial time, and by taking a closer look at the \drp{}-instance resulting from the reductions.
A \tsat-instance $\Phi$ is reduced to a \drp{}-instance $(\GG, s_1, f_n, \delta = 1, 2\cdot n-1)$ (with \mcpsatshrt{} as an intermediate problem).
A visualization of the resulting temporal graph $\GG$ can be seen in \cref{fig:3sat_drp_red_graph}. 
In each selection and finalization gadget there are two parallel paths starting and ending in the same vertex with two intermediate vertices for both parallel paths.
The validation gadget consist of chained gadgets for each clause.
For each clause gadget there are three parallel paths (one for each literal) starting and ending in the same vertex with two intermediate vertices per parallel path.
Hence, the bandwidth of the underlying graph $\under{\GG}$ is $3$.
Furthermore, the delay $\delta = 1$ and the maximum traversal time $\lambda_{\max} = \max_\lambda \{(v,w,t, \lambda) \in E\} = 0$.
\end{proof}

Next, we show \CWONE{}-hardness of \drp{} for the feedback vertex set of the underlying graph, the length of a delay-robust temporal path, and the number of delays combined.
To this end, we give a parameterized polynomial-time reduction from \mcc{}~\cite{MCC_W1hardness} to \drp{}.
Again we use \mcpsatshrt{} as an intermediate problem and use \cref{cor:poly_red_dj_drp}. Formally, we show the following result.

\begin{theorem} \label{theorem:drp_w1hardness}
	\drp{} is \CWONE{}-hard with respect to $x + L + f$
	where $x$ is the number of delays,
	$L$ is the length of a longest $s$-$z$ path in $\under{\GG}$,
	and $f$ is the feedback vertex number of $\under{\GG}$.
\end{theorem}

Given a $k$-partite graph (each partition of another color), \mcc{} asks whether there the graph contains clique of size $k$. \mcc{} is \CWONE{}-hard when parameterized by the number of partitions~$k$~\cite{MCC_W1hardness}. 
\problemdef{\mcc{}}
{A graph $G = (V,E)$ with $V = V_1 \uplus V_2 \uplus \ldots \uplus V_k$.}
{Is there a set of vertices $C \subseteq V$ so that $\forall v,w \in C: v\ne w \Rightarrow \{v,w\} \in E$ and $\abs{C} = k$?}

Let $V_i = \{v_{i,1}, v_{i,2}, \ldots, v_{i, \abs{V_i}}\}$.
Furthermore, let $E_{i,j}$ denote the set of edges between the partitions $V_i$ and $V_j$.
We construct a \mcpsatshrt{}-instance $((X_1, X_2, \ldots, X_k), \Phi)$ that is a yes-instance of \mcpsatshrt{} if and only if $G = (V_1 \uplus V_2 \uplus \ldots \uplus V_k, E)$ is a yes-instance of \mcc{}.
We define $X_i = \{x_{(i,j)} \mid v_{i,j} \in V_i\}$ to have a variable for each vertex in the partition $V_i$. 
Setting a variable $x_{(i,j)}$ to true corresponds to selecting this vertex for the clique. (Note that \mcpsatshrt{} asks for exactly one true variable in each set $X_i$ which ensures than only one vertex from each partition is selected.)
Now we construct $\Phi$ to ensure that in each pair of different partitions the selected vertices in those partitions are connected by an edge:
\[
	\Phi = \bigwedge_{1 \le i < j \le k} \left( \bigvee_{\{v_{i,a}, v_{j,b}\} \in E_{i,j}} x_{i,a} \wedge x_{j,b} \right)
\]

\begin{lemma} \label{lemma:mcc_to_djpsat}
	$(G = (V_1 \uplus V_2 \uplus \ldots \uplus V_k, E))$ is a yes-instance of \mcc{} if and only if $((X_1, X_2, \ldots, X_k), \Phi)$ is a yes-instance of \mcpsatshrt{}.
\end{lemma}
\begin{proof}
	~\vspace{-\baselineskip} %
	\proofsubparagraph{($\Rightarrow$):} Assume $(G = (V_1 \uplus V_2 \uplus \ldots \uplus V_k, E))$ is a yes-instance of \mcc{}.
	This means that there is a subset of vertices $C \subseteq V$ with $\abs{C} = k$ so that $\forall v,w \in C: v \ne w \implies \{v,w\} \in E$ holds.
	Since the partitions are pairwise disjoint and $\abs{C} = k$ there is exactly one vertex in $C$ from each partition. 
	Let this vertex be denoted as $v_{i,a_i}$ for each partition~$V_i$.
	For the \mcpsatshrt{}-instance we set the corresponding variable $x_{i,a_i}$ to true.
	It can be seen that $\Phi$ is satisfied.
	For all $1 \le i < j \le k,$ the disjunction 
	\[
		\bigvee_{\{v_{i,a}, v_{j,b}\} \in E_{i,j}} x_{i,a} \wedge x_{j,b}
	\]
	is satisfied, since there is an edge $\{v_{i,a}, v_{j,b}\} \in E_{i,j}$ (due to our assumption that we have a clique) and the corresponding variables $x_{i,a}$ and $x_{j,b}$ are both set to true. 

	\proofsubparagraph{($\Rightarrow$):} Assume $((X_1, X_2, \ldots, X_k), \Phi)$ is a yes-instance of \mcpsatshrt{}.
	Thus, there is a truth assignment that assigns exactly one variable from $X_i$ for $i \in [k]$ to true that satisfies $\Phi$.
	Let $x_{i,a}$ be the variable set to true in $X_i$. 
	We select the corresponding vertex $v_{i,a}$ to be in a set $C$.
	Since there are $k$ true variables, we have $\abs{C} = k$ and each vertex in $C$ from a different partition.
	Because $\Phi$ is satisfied for all $1 \le i < j \le k$ the disjunction 
	\[
		\bigvee_{\{v_{i,a}, v_{j,b}\} \in E_{i,j}} x_{i,a} \wedge x_{j,b}
	\]
	is satisfied.
	Hence, there exists one pair of variables $x_{i,a}$ and $x_{j,b}$ that are both true, and there is an edge $\{v_{i,a}, v_{j,b}\} \in E_{i,j}$ between the corresponding vertices.
	These vertices~$v_{i,a}$ and $v_{j,b}$ are both in $C$. 
	Thus, $C$ is a clique.
\end{proof}

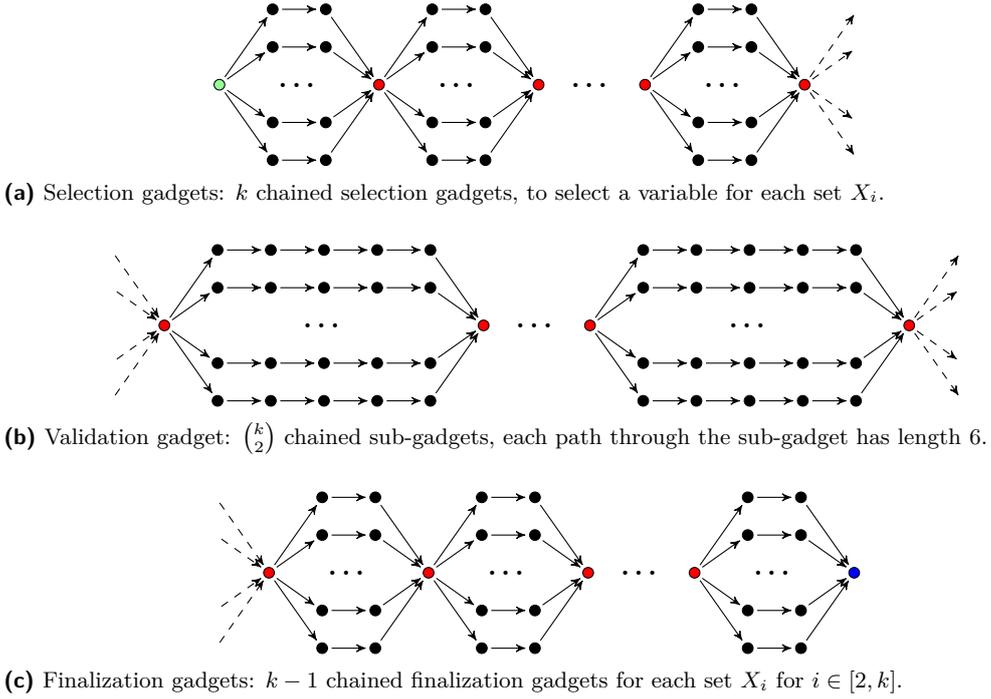
\begin{figure}[t]
	\centering
	\tikzstyle{alter}=[circle, minimum size=16pt, draw, inner sep=1pt] 
	\tikzstyle{majarr}=[draw=black]

	\begin{subfigure}[b]{\textwidth}
		\centering
		
		\begin{tikzpicture}[auto, >=stealth',shorten <=1pt, shorten >=1pt, state/.style={alter, scale=0.75, draw, minimum size=8cm}]
			\tikzstyle{majarr}=[draw=black,->,shorten <=1.5pt, shorten >=1.5pt]
			
\node[alter, fill={rgb,255:red,150; green,255; blue,150}, scale=0.25] at (0.0, 0) (sel1s) {};
\node[alter, fill=red, scale=0.25] at (2.0999999999999996, 0) (sel1z) {};
\node[alter, fill=black, scale=0.25] at (0.7, 1.0) (sel101) {};
\node[alter, fill=black, scale=0.25] at (1.4, 1.0) (sel102) {};
\draw[majarr] (sel1s) edge (sel101);
\draw[majarr] (sel101) edge (sel102);
\draw[majarr] (sel102) edge (sel1z);
\node[alter, fill=black, scale=0.25] at (0.7, 0.5) (sel111) {};
\node[alter, fill=black, scale=0.25] at (1.4, 0.5) (sel112) {};
\draw[majarr] (sel1s) edge (sel111);
\draw[majarr] (sel111) edge (sel112);
\draw[majarr] (sel112) edge (sel1z);
\node[alter, fill=black, scale=0.25] at (0.7, -0.5) (sel121) {};
\node[alter, fill=black, scale=0.25] at (1.4, -0.5) (sel122) {};
\draw[majarr] (sel1s) edge (sel121);
\draw[majarr] (sel121) edge (sel122);
\draw[majarr] (sel122) edge (sel1z);
\node[alter, fill=black, scale=0.25] at (0.7, -1.0) (sel131) {};
\node[alter, fill=black, scale=0.25] at (1.4, -1.0) (sel132) {};
\draw[majarr] (sel1s) edge (sel131);
\draw[majarr] (sel131) edge (sel132);
\draw[majarr] (sel132) edge (sel1z);
\path (sel1s) -- node[auto=false]{\textbf{\ldots}} (sel1z);
\node[alter, fill=red, scale=0.25] at (4.199999999999999, 0) (sel2z) {};
\node[alter, fill=black, scale=0.25] at (2.8, 1.0) (sel201) {};
\node[alter, fill=black, scale=0.25] at (3.5, 1.0) (sel202) {};
\draw[majarr] (sel1z) edge (sel201);
\draw[majarr] (sel201) edge (sel202);
\draw[majarr] (sel202) edge (sel2z);
\node[alter, fill=black, scale=0.25] at (2.8, 0.5) (sel211) {};
\node[alter, fill=black, scale=0.25] at (3.5, 0.5) (sel212) {};
\draw[majarr] (sel1z) edge (sel211);
\draw[majarr] (sel211) edge (sel212);
\draw[majarr] (sel212) edge (sel2z);
\node[alter, fill=black, scale=0.25] at (2.8, -0.5) (sel221) {};
\node[alter, fill=black, scale=0.25] at (3.5, -0.5) (sel222) {};
\draw[majarr] (sel1z) edge (sel221);
\draw[majarr] (sel221) edge (sel222);
\draw[majarr] (sel222) edge (sel2z);
\node[alter, fill=black, scale=0.25] at (2.8, -1.0) (sel231) {};
\node[alter, fill=black, scale=0.25] at (3.5, -1.0) (sel232) {};
\draw[majarr] (sel1z) edge (sel231);
\draw[majarr] (sel231) edge (sel232);
\draw[majarr] (sel232) edge (sel2z);
\path (sel1z) -- node[auto=false]{\textbf{\ldots}} (sel2z);
\node[alter, fill=red, scale=0.25] at (5.6, 0) (sel3s) {};
\node[alter, fill=red, scale=0.25] at (7.699999999999999, 0) (sel3z) {};
\node[alter, fill=black, scale=0.25] at (6.3, 1.0) (sel301) {};
\node[alter, fill=black, scale=0.25] at (7.0, 1.0) (sel302) {};
\draw[majarr] (sel3s) edge (sel301);
\draw[majarr] (sel301) edge (sel302);
\draw[majarr] (sel302) edge (sel3z);
\node[alter, fill=black, scale=0.25] at (6.3, 0.5) (sel311) {};
\node[alter, fill=black, scale=0.25] at (7.0, 0.5) (sel312) {};
\draw[majarr] (sel3s) edge (sel311);
\draw[majarr] (sel311) edge (sel312);
\draw[majarr] (sel312) edge (sel3z);
\node[alter, fill=black, scale=0.25] at (6.3, -0.5) (sel321) {};
\node[alter, fill=black, scale=0.25] at (7.0, -0.5) (sel322) {};
\draw[majarr] (sel3s) edge (sel321);
\draw[majarr] (sel321) edge (sel322);
\draw[majarr] (sel322) edge (sel3z);
\node[alter, fill=black, scale=0.25] at (6.3, -1.0) (sel331) {};
\node[alter, fill=black, scale=0.25] at (7.0, -1.0) (sel332) {};
\draw[majarr] (sel3s) edge (sel331);
\draw[majarr] (sel331) edge (sel332);
\draw[majarr] (sel332) edge (sel3z);
\path (sel3s) -- node[auto=false]{\textbf{\ldots}} (sel3z);
\path (sel2z) -- node[auto=false]{\textbf{\ldots}} (sel3s);
\node[scale=0.25] at (8.399999999999999, 1.0) (sel3zd0) {};
\draw[majarr, dashed] (sel3z) edge (sel3zd0);
\node[scale=0.25] at (8.399999999999999, 0.5) (sel3zd1) {};
\draw[majarr, dashed] (sel3z) edge (sel3zd1);
\node[scale=0.25] at (8.399999999999999, -0.5) (sel3zd2) {};
\draw[majarr, dashed] (sel3z) edge (sel3zd2);
\node[scale=0.25] at (8.399999999999999, -1.0) (sel3zd3) {};
\draw[majarr, dashed] (sel3z) edge (sel3zd3);

		\end{tikzpicture}	
		
		\caption{Selection gadgets: $k$ chained selection gadgets, to select a variable for each set $X_i$.}		
	\end{subfigure}
	
	\begin{tikzpicture}
		\node[] at (0,0) (a) {};
	\end{tikzpicture}

	\begin{subfigure}[b]{\textwidth}
		\centering
		
		\begin{tikzpicture}[auto, >=stealth',shorten <=1pt, shorten >=1pt, state/.style={alter, scale=0.75, draw, minimum size=8cm}]
			\tikzstyle{majarr}=[draw=black,->,shorten <=1.5pt, shorten >=1.5pt]
			
\node[alter, fill=red, scale=0.25] at (0.0, 0) (val1s) {};
\node[alter, fill=red, scale=0.25] at (4.199999999999999, 0) (val1z) {};
\node[alter, fill=black, scale=0.25] at (0.7, 1.0) (val101) {};
\node[alter, fill=black, scale=0.25] at (1.4, 1.0) (val102) {};
\node[alter, fill=black, scale=0.25] at (2.0999999999999996, 1.0) (val103) {};
\node[alter, fill=black, scale=0.25] at (2.8, 1.0) (val104) {};
\node[alter, fill=black, scale=0.25] at (3.5, 1.0) (val105) {};
\draw[majarr] (val1s) edge (val101);
\draw[majarr] (val101) edge (val102);
\draw[majarr] (val102) edge (val103);
\draw[majarr] (val103) edge (val104);
\draw[majarr] (val104) edge (val105);
\draw[majarr] (val105) edge (val1z);
\node[alter, fill=black, scale=0.25] at (0.7, 0.5) (val111) {};
\node[alter, fill=black, scale=0.25] at (1.4, 0.5) (val112) {};
\node[alter, fill=black, scale=0.25] at (2.0999999999999996, 0.5) (val113) {};
\node[alter, fill=black, scale=0.25] at (2.8, 0.5) (val114) {};
\node[alter, fill=black, scale=0.25] at (3.5, 0.5) (val115) {};
\draw[majarr] (val1s) edge (val111);
\draw[majarr] (val111) edge (val112);
\draw[majarr] (val112) edge (val113);
\draw[majarr] (val113) edge (val114);
\draw[majarr] (val114) edge (val115);
\draw[majarr] (val115) edge (val1z);
\node[alter, fill=black, scale=0.25] at (0.7, -0.5) (val121) {};
\node[alter, fill=black, scale=0.25] at (1.4, -0.5) (val122) {};
\node[alter, fill=black, scale=0.25] at (2.0999999999999996, -0.5) (val123) {};
\node[alter, fill=black, scale=0.25] at (2.8, -0.5) (val124) {};
\node[alter, fill=black, scale=0.25] at (3.5, -0.5) (val125) {};
\draw[majarr] (val1s) edge (val121);
\draw[majarr] (val121) edge (val122);
\draw[majarr] (val122) edge (val123);
\draw[majarr] (val123) edge (val124);
\draw[majarr] (val124) edge (val125);
\draw[majarr] (val125) edge (val1z);
\node[alter, fill=black, scale=0.25] at (0.7, -1.0) (val131) {};
\node[alter, fill=black, scale=0.25] at (1.4, -1.0) (val132) {};
\node[alter, fill=black, scale=0.25] at (2.0999999999999996, -1.0) (val133) {};
\node[alter, fill=black, scale=0.25] at (2.8, -1.0) (val134) {};
\node[alter, fill=black, scale=0.25] at (3.5, -1.0) (val135) {};
\draw[majarr] (val1s) edge (val131);
\draw[majarr] (val131) edge (val132);
\draw[majarr] (val132) edge (val133);
\draw[majarr] (val133) edge (val134);
\draw[majarr] (val134) edge (val135);
\draw[majarr] (val135) edge (val1z);
\path (val1s) -- node[auto=false]{\textbf{\ldots}} (val1z);
\node[alter, fill=red, scale=0.25] at (5.6, 0) (sel3s) {};
\node[alter, fill=red, scale=0.25] at (9.799999999999999, 0) (sel3z) {};
\node[alter, fill=black, scale=0.25] at (6.3, 1.0) (sel301) {};
\node[alter, fill=black, scale=0.25] at (7.0, 1.0) (sel302) {};
\node[alter, fill=black, scale=0.25] at (7.699999999999999, 1.0) (sel303) {};
\node[alter, fill=black, scale=0.25] at (8.399999999999999, 1.0) (sel304) {};
\node[alter, fill=black, scale=0.25] at (9.1, 1.0) (sel305) {};
\draw[majarr] (sel3s) edge (sel301);
\draw[majarr] (sel301) edge (sel302);
\draw[majarr] (sel302) edge (sel303);
\draw[majarr] (sel303) edge (sel304);
\draw[majarr] (sel304) edge (sel305);
\draw[majarr] (sel305) edge (sel3z);
\node[alter, fill=black, scale=0.25] at (6.3, 0.5) (sel311) {};
\node[alter, fill=black, scale=0.25] at (7.0, 0.5) (sel312) {};
\node[alter, fill=black, scale=0.25] at (7.699999999999999, 0.5) (sel313) {};
\node[alter, fill=black, scale=0.25] at (8.399999999999999, 0.5) (sel314) {};
\node[alter, fill=black, scale=0.25] at (9.1, 0.5) (sel315) {};
\draw[majarr] (sel3s) edge (sel311);
\draw[majarr] (sel311) edge (sel312);
\draw[majarr] (sel312) edge (sel313);
\draw[majarr] (sel313) edge (sel314);
\draw[majarr] (sel314) edge (sel315);
\draw[majarr] (sel315) edge (sel3z);
\node[alter, fill=black, scale=0.25] at (6.3, -0.5) (sel321) {};
\node[alter, fill=black, scale=0.25] at (7.0, -0.5) (sel322) {};
\node[alter, fill=black, scale=0.25] at (7.699999999999999, -0.5) (sel323) {};
\node[alter, fill=black, scale=0.25] at (8.399999999999999, -0.5) (sel324) {};
\node[alter, fill=black, scale=0.25] at (9.1, -0.5) (sel325) {};
\draw[majarr] (sel3s) edge (sel321);
\draw[majarr] (sel321) edge (sel322);
\draw[majarr] (sel322) edge (sel323);
\draw[majarr] (sel323) edge (sel324);
\draw[majarr] (sel324) edge (sel325);
\draw[majarr] (sel325) edge (sel3z);
\node[alter, fill=black, scale=0.25] at (6.3, -1.0) (sel331) {};
\node[alter, fill=black, scale=0.25] at (7.0, -1.0) (sel332) {};
\node[alter, fill=black, scale=0.25] at (7.699999999999999, -1.0) (sel333) {};
\node[alter, fill=black, scale=0.25] at (8.399999999999999, -1.0) (sel334) {};
\node[alter, fill=black, scale=0.25] at (9.1, -1.0) (sel335) {};
\draw[majarr] (sel3s) edge (sel331);
\draw[majarr] (sel331) edge (sel332);
\draw[majarr] (sel332) edge (sel333);
\draw[majarr] (sel333) edge (sel334);
\draw[majarr] (sel334) edge (sel335);
\draw[majarr] (sel335) edge (sel3z);
\path (sel3s) -- node[auto=false]{\textbf{\ldots}} (sel3z);
\path (val1z) -- node[auto=false]{\textbf{\ldots}} (sel3s);
\node[scale=0.25] at (-0.7, 1.0) (val1sd0) {};
\draw[majarr, dashed] (val1sd0) edge (val1s);
\node[scale=0.25] at (-0.7, 0.5) (val1sd1) {};
\draw[majarr, dashed] (val1sd1) edge (val1s);
\node[scale=0.25] at (-0.7, -0.5) (val1sd2) {};
\draw[majarr, dashed] (val1sd2) edge (val1s);
\node[scale=0.25] at (-0.7, -1.0) (val1sd3) {};
\draw[majarr, dashed] (val1sd3) edge (val1s);
\node[scale=0.25] at (10.5, 1.0) (sel3zd0) {};
\draw[majarr, dashed] (sel3z) edge (sel3zd0);
\node[scale=0.25] at (10.5, 0.5) (sel3zd1) {};
\draw[majarr, dashed] (sel3z) edge (sel3zd1);
\node[scale=0.25] at (10.5, -0.5) (sel3zd2) {};
\draw[majarr, dashed] (sel3z) edge (sel3zd2);
\node[scale=0.25] at (10.5, -1.0) (sel3zd3) {};
\draw[majarr, dashed] (sel3z) edge (sel3zd3);

		\end{tikzpicture}	
		
		\caption{Validation gadget: $\binom{k}{2}$ chained sub-gadgets, each path through the sub-gadget has length $6$.}		
	\end{subfigure}
	
	\begin{tikzpicture}
		\node[] at (0,0) (a) {};
	\end{tikzpicture}

	\begin{subfigure}[b]{\textwidth}
		\centering
		
		\begin{tikzpicture}[auto, >=stealth',shorten <=1pt, shorten >=1pt, state/.style={alter, scale=0.75, draw, minimum size=8cm}]
			\tikzstyle{majarr}=[draw=black,->,shorten <=1.5pt, shorten >=1.5pt]
			
\node[alter, fill=red, scale=0.25] at (0.0, 0) (fin1s) {};
\node[alter, fill=red, scale=0.25] at (2.0999999999999996, 0) (fin1z) {};
\node[alter, fill=black, scale=0.25] at (0.7, 1.0) (fin101) {};
\node[alter, fill=black, scale=0.25] at (1.4, 1.0) (fin102) {};
\draw[majarr] (fin1s) edge (fin101);
\draw[majarr] (fin101) edge (fin102);
\draw[majarr] (fin102) edge (fin1z);
\node[alter, fill=black, scale=0.25] at (0.7, 0.5) (fin111) {};
\node[alter, fill=black, scale=0.25] at (1.4, 0.5) (fin112) {};
\draw[majarr] (fin1s) edge (fin111);
\draw[majarr] (fin111) edge (fin112);
\draw[majarr] (fin112) edge (fin1z);
\node[alter, fill=black, scale=0.25] at (0.7, -0.5) (fin121) {};
\node[alter, fill=black, scale=0.25] at (1.4, -0.5) (fin122) {};
\draw[majarr] (fin1s) edge (fin121);
\draw[majarr] (fin121) edge (fin122);
\draw[majarr] (fin122) edge (fin1z);
\node[alter, fill=black, scale=0.25] at (0.7, -1.0) (fin131) {};
\node[alter, fill=black, scale=0.25] at (1.4, -1.0) (fin132) {};
\draw[majarr] (fin1s) edge (fin131);
\draw[majarr] (fin131) edge (fin132);
\draw[majarr] (fin132) edge (fin1z);
\path (fin1s) -- node[auto=false]{\textbf{\ldots}} (fin1z);
\node[alter, fill=red, scale=0.25] at (4.199999999999999, 0) (fin2z) {};
\node[alter, fill=black, scale=0.25] at (2.8, 1.0) (fin201) {};
\node[alter, fill=black, scale=0.25] at (3.5, 1.0) (fin202) {};
\draw[majarr] (fin1z) edge (fin201);
\draw[majarr] (fin201) edge (fin202);
\draw[majarr] (fin202) edge (fin2z);
\node[alter, fill=black, scale=0.25] at (2.8, 0.5) (fin211) {};
\node[alter, fill=black, scale=0.25] at (3.5, 0.5) (fin212) {};
\draw[majarr] (fin1z) edge (fin211);
\draw[majarr] (fin211) edge (fin212);
\draw[majarr] (fin212) edge (fin2z);
\node[alter, fill=black, scale=0.25] at (2.8, -0.5) (fin221) {};
\node[alter, fill=black, scale=0.25] at (3.5, -0.5) (fin222) {};
\draw[majarr] (fin1z) edge (fin221);
\draw[majarr] (fin221) edge (fin222);
\draw[majarr] (fin222) edge (fin2z);
\node[alter, fill=black, scale=0.25] at (2.8, -1.0) (fin231) {};
\node[alter, fill=black, scale=0.25] at (3.5, -1.0) (fin232) {};
\draw[majarr] (fin1z) edge (fin231);
\draw[majarr] (fin231) edge (fin232);
\draw[majarr] (fin232) edge (fin2z);
\path (fin1z) -- node[auto=false]{\textbf{\ldots}} (fin2z);
\node[alter, fill=red, scale=0.25] at (5.6, 0) (fin3s) {};
\node[alter, fill=blue, scale=0.25] at (7.699999999999999, 0) (fin3z) {};
\node[alter, fill=black, scale=0.25] at (6.3, 1.0) (fin301) {};
\node[alter, fill=black, scale=0.25] at (7.0, 1.0) (fin302) {};
\draw[majarr] (fin3s) edge (fin301);
\draw[majarr] (fin301) edge (fin302);
\draw[majarr] (fin302) edge (fin3z);
\node[alter, fill=black, scale=0.25] at (6.3, 0.5) (fin311) {};
\node[alter, fill=black, scale=0.25] at (7.0, 0.5) (fin312) {};
\draw[majarr] (fin3s) edge (fin311);
\draw[majarr] (fin311) edge (fin312);
\draw[majarr] (fin312) edge (fin3z);
\node[alter, fill=black, scale=0.25] at (6.3, -0.5) (fin321) {};
\node[alter, fill=black, scale=0.25] at (7.0, -0.5) (fin322) {};
\draw[majarr] (fin3s) edge (fin321);
\draw[majarr] (fin321) edge (fin322);
\draw[majarr] (fin322) edge (fin3z);
\node[alter, fill=black, scale=0.25] at (6.3, -1.0) (fin331) {};
\node[alter, fill=black, scale=0.25] at (7.0, -1.0) (fin332) {};
\draw[majarr] (fin3s) edge (fin331);
\draw[majarr] (fin331) edge (fin332);
\draw[majarr] (fin332) edge (fin3z);
\path (fin3s) -- node[auto=false]{\textbf{\ldots}} (fin3z);
\path (fin2z) -- node[auto=false]{\textbf{\ldots}} (fin3s);
\node[scale=0.25] at (-0.7, 1.0) (fin1sd0) {};
\draw[majarr, dashed] (fin1sd0) edge (fin1s);
\node[scale=0.25] at (-0.7, 0.5) (fin1sd1) {};
\draw[majarr, dashed] (fin1sd1) edge (fin1s);
\node[scale=0.25] at (-0.7, -0.5) (fin1sd2) {};
\draw[majarr, dashed] (fin1sd2) edge (fin1s);
\node[scale=0.25] at (-0.7, -1.0) (fin1sd3) {};
\draw[majarr, dashed] (fin1sd3) edge (fin1s);

		\end{tikzpicture}		
		
		\caption{Finalization gadgets: $k-1$ chained finalization gadgets for each set $X_i$ for $i\in [2,k]$.}
	\end{subfigure}
	
	\caption{The temporal graph resulting from the \mcc{} $\lepoly$ \drp{} reduction. 
	The start vertex is highlighted light green, the end vertex is highlighted dark blue. 
	All time labels are ommited. 
	Red vertices form a feedback vertex set of the underlying temporal graph.}
	\label{fig:mcc_drp_red_graph}
\end{figure}

\begin{proof}[Proof of \cref{theorem:drp_w1hardness}]
To obtain \cref{theorem:drp_w1hardness} from \cref{lemma:mcc_to_djpsat} we will look at the \drp{}-instance $(\GG, s_1, f_k, \delta=1, 2\cdot k - 1)$ that results from the reduction. First, observe that the reduction can be performed in polynomial time.
We can immediately see that the number of allowed delays $2\cdot k - 1$ is upper-bounded by a function depending only in $k$.
The temporal graph $\GG$ is visualized in \cref{fig:mcc_drp_red_graph}.
In this figure a set of feedback vertices is highlighted in cyan.
One can see that the number of cyan feedback vertices is $2\cdot k + \binom{k}{2} - 2$: 
For each of the $k$ selection gadgets and each of the $k-1$ finalization gadgets there is one cyan vertex.
For each of the $\binom{k}{2}$ chained sub-gadgets in the validation gadget there is one cyan vertex plus one additional cyan vertex.
Two vertices where double counted since the start and end of the validation gadget where already counted in the selection and finalization gadget respectively.
Hence, the feedback vertex number is upper-bounded by a function in $k$.
Furthermore, we can see that all possible paths from $s_1$ (green vertex) to $f_k$ (red vertex) have a fixed length of $3 \cdot (2k-1) + 6 \cdot \binom{k}{2}$ which is also upper-bounded by a function in~$k$.
\end{proof}

The presented hardness result show that we presumably cannot generalize \cref{thm:polyforest} to an FPT-result for parameters such as the treewidth of the underlying graph or the feedback vertex number of the underlying graph.

\section{Parameterized Algorithms}\label{sec:algs}

In \cref{subsec:reduction_framework}, we presented several hardness results. Here, we present our algorithmic results for general input graphs which can be seen as different ways to generalize \cref{thm:polyforest}. We start with an XP-algorithm for the number of delays as a parameter and then present two FPT-algorithms for ``distance to forest'' parameters.

\subsection{Number of Delays} \label{sec:drp_fpt_delay_tau}

In what follows, we present an algorithm similar to Dijkstra's algorithm~\cite{dijkstra1959note}. 
Starting at the source vertex $s$, it finds all optimal temporal $(s,v)$-routes by expanding each optimum by one step per iteration.
However, as we have seen in the polynomial-time reductions in \cref{subsec:reduction_framework}, there can be many $(s,z)$-routes that are Pareto-optimal with respect to the arrival time for a given number of delays.
We use the dynamic program from \cref{thm:polyforest} to extend the paths by a single time arc.
Our main result of this section is that \drp{} admits an \CXP{}-algorithm with respect to the number $x$ of delays. \cref{theorem:drp_w1hardness} implies that we presumably cannot improve this to an FPT result for this parameter.
Formally, we show the following.

\begin{theorem} \label{theorem:drp_xp_x}
	\drp{} can be solved in 
	\(
	\mathcal{O}(\abs{V}^3 \cdot \abs{E}^{2x} \cdot x^2)
	\)
	time, where $x$ is the number of allowed delays.
\end{theorem}

For each route, its \emph{arrival time vector} $\vec{t} = (t_0, t_1, \ldots, t_x)$ 
is a vector of $x+1$ time steps
where $t_y$ is the worst-case arrival time for $y$ delays.
We define a partial order~$\preceq$ to compare arrival time vectors.
	For $\vec{t} = (t_0, t_1, \ldots, t_x)$ and $\vec{t'} = (t'_0, t'_1, \ldots, t'_x)$, set
	$\vec{t} \preceq \vec{t'}$ if and only if  $t_y \le t'_y$ for all $y \leq x$.
This partial order can be used to decrease the set of prefix paths that need to be considered due to the following observation.

\begin{observation} \label{obs:preceq_arrival_time}
	Let $\GG = (V,E)$ be a temporal graph, let $s,v,z \in V$ be three vertices, and $P_1$ and $P_2$ be two delay-robust $(s,v)$-routes
	with the arrival time vectors $\vec{t}_1 \preceq \vec{t}_2$. 
	If there is a delay-robust $(s,z)$-route~$P$ so that $P = P_2 \circ P'$,
	then $P_1 \circ P'$ is also a delay-robust route.
	Additionally, if there is a delay-robust $(s,v)$-route,
	then there is one whose arrival time vector is minimal among all $(s,v)$-routes.
\end{observation}

Since for any delay one can arrive earlier in vertex~$v$ by using the route~$P_1$ compared to~$P_2$,
replacing the prefix~$P_2$ by~$P_1$ still guarantees delay-robustness.

We define a table~$A$ with entries for every vertex of $\GG$. The table entry $A[v]$ contains a set of arrival time vectors for $(s,v)$-routes. 
We will only store vectors that are minimal with respect to $\preceq$, since we do not need to consider others due to \cref{obs:preceq_arrival_time}.
Thus, the set $A[v]$ will represent the Pareto front of routes from~$s$ to~$v$.

Furthermore, we define a priority queue $Q$ that contains tuples~$(v,\vec{t})$ of vertices and arrival time vectors.
The queue is sorted according by the arrival time vectors according to $\preceq$. 
The queue elements $(v, \vec{t})$ contain the prefix routes from where a search should be expanded.

We initialize the table~$A$ as follows:
\[
	A[v] = 
	\begin{cases}
		\{ (0, \ldots, 0) \}, & \text{if } v = s\\
		\emptyset, & \text{otherwise.}
	\end{cases}
\]

The start vertex $s$ can always be reached through the empty path. 
For all other vertices there is initially no route stored.
Furthermore, we initialize the queue~$Q$ with the tuple $(s, (0, \ldots, 0))$. 

To compute the table entries we repeatedly pop the first element $(v, \vec{t})$ from $Q$ and propagate possible delay-robust routes from there.
If $(v, \vec{t})$ is in the queue, then this means that there is a delay-robust $(s,v)$-route~$P$ with the arrival time vector~$\vec{t}$. 

Let $\text{next}_v := \{ w \mid (v,w,t,\lambda) \in V \}$ denote the set of vertices reachable from $v$ by a single time arc.
For all $w \in \text{next}_v$, we compute the arrival time vector~$\vec{t'} = (t'_0, t'_1, \dots, t'_x)$ of~$P' = P \circ (w)$ using the dynamic program described in \cref{subsec:drp_verification}:
The arrival time vector of~$P'$ is simply the table row $A_{P'}$
and $P'$ is $y$-delay-robust if and only if $A_{P'}[y] < \infty$.

As an optimization, we can round up the arrival time entries to the next time step in $\tau_{w}^+$,
i.e. replace $t'_y$ by 
\[ \hat{t'_y} = \min_{t} \{ t \in \tau_{w}^+  \mid t \ge t'_y \}\].
This rounding does not change the delay-robustness of any route since no temporal walk can leave~$w$ between time~$t'_y$ and $\hat{t'_y}$.

If $P'$ is $x$-delay-robust, then we can add $\vec{t'}$ to the set $A[w]$, unless $A[w]$ already contains a smaller arrival time vector.
We then delete all $\vec{t''}$ with $\vec{t'} \preceq \vec{t''}$ from~$A[w]$ and also remove the corresponding elements $(w, \vec{t''})$ from the queue $Q$.
Finally, we insert $(w,\vec{t'})$ into~$Q$.

Once the queue~$Q$ is empty, we have investigated all $x$-robust prefix routes that might eventually lead to~$z$.
There is then a delay-robust $(s,v)$-route if and only if $A[z] \neq \emptyset$.

To analyze the running time we need to count the number of items of the queue being processed.

\begin{lemma}
\label{lemma:drp_fpt_x_time}
	Let $\abs{\tau_{\max}^+} = \max_{v \in V} \abs{\tau_{v}^+}$ be the maximum number of distinct time steps, where a single vertex has outgoing time arcs.
	The algorithm has a running time of 
	$$\mathcal{O}\left(\abs{V}^2 \cdot \abs{\tau_{\max}^+}^x \cdot \left(\abs{E} \cdot x^2 + \abs{V} \cdot \abs{\tau_{\max}^+}^x \cdot x\right) + \abs{E} \cdot \log \abs{E}\right).$$
\end{lemma}
\begin{proof}
	Since we round the arrival time up to the next largest time step, where a vertex has outgoing time arcs, the earliest arrival time for $x$ delays can be one of at most $\abs{\tau_{v}^+} \le \abs{\tau_{\text{max}}^+}$ distinct arrival times.
	Hence, for any vertex $v$, there are at most $\abs{\tau_{v}^+}^x$ distinct arrival time vectors and therefore $\abs{A[v]} \in \mathcal{O}(\abs{\tau_{v}^+}^x)$.
	Furthermore, $\abs{Q} \in \mathcal{O}(\abs{V} \cdot \abs{\tau_{v}^+}^x)$.
	
	Given a delay-robust path $P = (s, \ldots, v)$ with arrival time vector $\vec{t}$, the single step computation whether for a successor vertex $w$ the path $P \circ (w)$ is delay-robust, is done the same way as in the verification algorithm in \cref{subsec:drp_verification}. 
	Hence, the time arcs are first sorted with respect to the arrival time in $\mathcal{O}(\abs{E} \cdot \log \abs{E})$ time.
	The single step computation itself takes $\mathcal{O}(\abs{E} \cdot x^2)$ time (see \cref{lemma:drp_verif_running_time}).
	This single step computation is done for all successor vertices of $v$, which are at most $\abs{V}$ many.
	Having computed a delay-robust path $P \circ (w)$ with arrival time vector $\vec{t'}$ we need to check whether we can add it to $A[w]$ and the priority queue.
	To do so, all (at most $\abs{\tau_{\text{max}}^+}^x$) entries in $A[w]$ are checked against $\vec{t'}$ which takes $\mathcal{O}(\abs{\tau_{\text{max}}^+}^x \cdot x)$ time.
	If the found path is a Pareto optimum, then it can be added to the priority queue and worse items can be removed in $\mathcal{O}(\abs{V} \cdot \abs{\tau_{\text{max}}^+}^x \cdot x)$ time (if naively implemented as a linked list).
	Thus, for any pair $(v,\vec{t})$ popped from the priority queue the algorithm consumes
	\begin{align*}
		& \mathcal{O}\left(\abs{V} \cdot \left(\abs{E} \cdot x^2 + \abs{\tau_{\text{max}}^+}^x \cdot x + \abs{V} \cdot \abs{\tau_{\text{max}}^+}^x \cdot x\right)\right) \\
		={}& \mathcal{O}\left(\abs{V} \cdot \left(\abs{E} \cdot x^2 + \abs{V} \cdot \abs{\tau_{\text{max}}^+}^x \cdot x\right)\right)
	\end{align*}
	time steps.
	
	For any path $P = (s, \ldots, v)$ with arrival time vector $\vec{t}$ and a computed $P \circ (w)$ with arrival time vector $\vec{t'}$ we have $\vec{t} \preceq \vec{t'}$.
	Hence, the front elements of the priority queue are always Pareto optima and a total of at most $\abs{V} \cdot \abs{\tau_{\text{max}}^+}^x$ elements are processed in the single step routine.
	
	This gives an overall running time of 
	\[
		\mathcal{O}\left(\abs{V}^2 \cdot \abs{\tau_{\text{max}}^+}^x \cdot \left(\abs{E} \cdot x^2 + \abs{V} \cdot \abs{\tau_{\text{max}}^+}^x \cdot x\right) + \abs{E} \cdot\log \abs{E}\right).
	\]	
\end{proof}

The correctness of the algorithm follows from \cref{lemma:drw-table-correct} and \cref{obs:preceq_arrival_time}.

\begin{lemma}
\label{lem:XPcorrectness}
	The described algorithm solves \drp{}.
\end{lemma}
\begin{proof}
	Given a route $P = (s, \ldots, v)$, we compute for a successor vertex $w$ of $v$ whether the route $P \circ (w) = (s, \ldots, v, w)$ is delay-robust for any delay up to size $x$ and its associated arrival time vector.
	This is done by the dynamic program introduced in \cref{subsec:drp_verification} and its correctness is proven in \cref{lemma:drw-table-correct}.
	This is done for all successor vertices of $v$.
	Any newly found delay-robust path is added to the priority queue $Q$, if there is no better path with respect to $\preceq$ and the same end vertex.
	Additionally, no item $(v,\vec{t})$ is removed from the queue without computing paths for the successor vertices of $v$, if there is no better path with respect to $\preceq$ and the end vertex $v$.
	Hence, all delay-robust routes from $s$ are found that are optimal with respect to $\preceq$.
	Due to \cref{obs:preceq_arrival_time}, this is enough to find out if there is a delay-robust path from vertex $s$ to any vertex $v$. 
\end{proof}

\cref{theorem:drp_xp_x} now directly follows from \cref{lemma:drp_fpt_x_time} and \cref{lem:XPcorrectness} together with the observation that $\abs{\tau_{\text{max}}^+} \le \abs{E}$.

\subsection{Timed Feedback Vertex Number}
In this section, we explore another way to generalize \cref{thm:polyforest}. We present an FPT algorithm for the so-called \emph{timed feedback vertex number} (introduced by Casteigts et al.~\cite{TemporalPathsUnderWaitingTime}) and the number $x$ of delays combined. Intuitively, the timed feedback vertex number is the minimum number of ``vertex appearances'' that need to be removed from the temporal graph to turn its underlying graph into a forest. Formally, it is defined as follows.

 Let
$\GG$ be a temporal graph and $X\subseteq V\times[T]$ a set of
\emph{vertex appearances}. Then we write $\GG-X:= (V, E')$,
where $E'=E\setminus \{(v,w,t,\lambda)\mid (v,t)\in
X \lor (w,t)\in X\}$.
A \emph{timed feedback vertex set} of $\GG$ is a set $X\subseteq V\times[T]$ of vertex appearances such that $\under{\GG-X}$ is cycle-free.
The \emph{timed feedback vertex number} of a temporal graph~$\GG$ is the minimum cardinality of a timed feedback vertex set of $\GG$.

\begin{theorem} \label{thm:tfvn}
	\drp{} can be solved in $2^{\bigO(x f \log f)} \cdot (\abs{V} + \abs{E})^{\bigO(1)}$ time, where $f$ is the timed feedback vertex number of the underlying graph.
\end{theorem}
In the following, we give a description of the main steps of the algorithm we use to obtain the above result. The algorithm follows a simple ``guess and check''-approach.
\newcommand{\relevant}{\hat{T}}
\begin{enumerate}
\item Compute a minimum timed feedback vertex set $X$ of the input graph using an algorithm provided by Casteigts et al.~\cite{TemporalPathsUnderWaitingTime}.
\item Let $\hat{X}=\{v\mid (v,t)\in X\}$. Iterate over all partitions $\hat{X}_0\uplus\hat{X}_1\uplus\hat{X}_2\uplus\hat{X}_3=\hat{X}$ of $\hat{X}$.
We distinguish two types of neighbors of a vertex.
A neighbor connected by a time arc that is preserved in $\GG - X$ is called a ``forest neighbor'',
while other neighbors are called ``feedback neighbors''.
Intuitively, in this step we guess for each vertex whether its predecessor resp.\ successor in the route is a feedback neighbor or a forest neighbor, leading to the following four cases:
\begin{itemize}
\item The route does not contain $v$ or the predecessor and successor of $v$ in the route are forest neighbors of $v$ (then $v\in\hat{X}_0$), 
\item the predecessor of $v$ in the route is a forest neighbor $v$, and the successor of $v$ in the route is a feedback neighbor of $v$  (then $v\in\hat{X}_1$),
\item the predecessor of $v$ in the route is a feedback neighbor $v$, and the successor of $v$ in the route is a forest neighbor of $v$ (then $v\in\hat{X}_2$), or
\item the predecessor and successor of $v$ in the route are feedback neighbors of $v$ (then $v\in\hat{X}_3$).
\end{itemize}
\item Iterate over all orders on $\hat{X}_1\cup\hat{X}_2\cup\hat{X}_3$. Intuitively, in this step we guess in which order the vertices appear in the route.
\item Let $\relevant =\{t, t+\delta \mid \exists w\in V : (w,t)\in X\} \cup \{\infty\}$ be the \emph{relevant} time steps.
For each vertex $v \in \hat{X}_1\cup\hat{X}_2\cup\hat{X}_3$,
iterate over all \emph{delay profiles} $(t_1, t_2, \ldots, t_x)\in \relevant^{x}$.
Intuitively, here we guess for each delay size~$i$
the smallest relevant time~$t_i$ which is at least the worst-case arrival time at~$v$.
\item Use \cref{thm:polyforest} to find route segments that respect the guessed delay profiles between consecutive vertices in $\hat{X}_1\cup\hat{X}_2\cup\hat{X}_3$ 
and which can be combined to an $x$-delay-robust $(s,z)$-route. 
\end{enumerate}

We first give a more detailed description of the last step of the algorithm.
Let $\relevant=\{t,t+\delta \mid \exists v\in V : (v,t)\in X\} \cup \{\infty\}$ be the set of relevant time steps.
Let $\{v_1,v_2,\ldots, v_{f'}\}=\hat{X}_1\uplus\hat{X}_2\uplus\hat{X}_3$ be the order of the vertices of the current iteration
and $\vec{v}_1,\vec{v}_2,\ldots, \vec{v}_{f'}\in \relevant^{x}$ the corresponding delay profiles.
For a vertex $v\in \hat{X}_1\uplus\hat{X}_2\uplus\hat{X}_3$,
the \emph{delay profile} $\vec{v}=(t_1, t_2, \ldots, t_x)\in  \relevant^{x}$
specifies, for each delay size~$i$,
the earliest relevant time step~$t_i$ that upper-bounds the worst-case arrival time at~$v$.

Now we want to be able to check for every consecutive pair of vertices $v_i,v_{i+1}$ with $i\in[f'-1]$ whether there is a route connecting them that respects both of their delay profiles. 
To do this, we use a subroutine to solve the following problem: 

Given an $(s,z)$-route~$R$
and delay profiles $\vec{s}=(t^{(s)}_1, t^{(s)}_2, \ldots, t^{(s)}_{x})$,
$\vec{z}=(t^{(z)}_1, t^{(z)}_2, \ldots, t^{(z)}_{x})$,
is it true for all~$j$ that $t_j$ is the smallest relevant time step
which, for all~$i \leq j$,
upper-bounds the worst-case arrival time at~$z$ for~$j-i$~delays
when starting at~$s$ at time~$t_i$?

It is easy to observe that this problem can be solved with $O(f^2)$ applications of \cref{thm:polyforest}.
From now on, let $\text{checkRoute}(R,\vec{s},\vec{z})\rightarrow \{\texttt{true},\texttt{false}\}$ denote a subroutine answering the above question. 

Now we go back to our original problem of checking whether there is a route connecting the vertices in $\hat{X}_1\uplus\hat{X}_2\uplus\hat{X}_3$.
For each $i \leq f' - 1$ there are four possible cases:
\begin{itemize}
\item $v_i\in\hat{X}_1\cup\hat{X}_3$ and $v_{i+1}\in\hat{X}_1$:
This means that the successor of $v_i$ is a feedback set neighbor and the predecessor of $v_{i+1}$ is a forest neighbor.

Let $N :=\{v\mid \exists t, \lambda: (v_i,v,t,\lambda)\in E\wedge (v_i,t)\in X\}$.
For each $u\in N$,
there is a unique $(u, v_{i+1})$-route~$R'$ in $\GG - X$.
Construct $R  = (v_i) \circ R'$
and check whether $\text{checkRoute}(R,\vec{v}_i,\vec{v}_{i+1})=\texttt{true}$.

\item $v_i\in\hat{X}_1\cup\hat{X}_3$ and $v_{i+1}\in\hat{X}_2\cup\hat{X}_3$:
This means that the successor of $v_i$ is a feedback neighbor and the predecessor of $v_{i+1}$ is also a feedback neighbor.

Let $N:=\{v\mid \exists t, \lambda: (v_i,v,t,\lambda)\in E\wedge (v_i,t)\in X\}$
and $N' :=\{v\mid \exists t, \lambda: (v,v_{i+1},t,\lambda)\in E\wedge (v_{i+1},t)\in X\}$.
For each $(u,u')\in N \times N'$,
there is a unique $(u,u')$-route~$R'$ in $\GG - X$.
Set $R = (v_i) \circ R' \circ (v_{i+1})$
and check whether $\text{checkRoute}(R,\vec{v}_i,\vec{v}_{i+1})=\texttt{true}$.

\item $v_i\in\hat{X}_2$ and $v_{i+1}\in\hat{X}_1$:
This means that the successor of $v_i$ is a forest neighbor and the predecessor of $v_{i+1}$ is also a forest neighbor.

There is a unique $(v_i, v_{i+1})$-route~$R$ in $\GG - X$.
We check whether $\text{checkRoute}(R,\vec{v}_i,\vec{v}_{i+1})=\texttt{true}$.

\item $v_i\in\hat{X}_2$ and $v_{i+1}\in\hat{X}_2\cup\hat{X}_3$:
This means that the successor of $v_i$ is a forest neighbor and the predecessor of $v_{i+1}$ is a feedback neighbor.

Let $N:=\{v\mid \exists t, \lambda: (v,v_{i+1},t,\lambda)\in E\wedge (v_{i+1},t)\in X\}$.
For each $u \in N$, there is a unique $(v_i, u)$-route~$R'$ in $\GG - X$.
Construct from this the route $R = R' \circ (v_{i+1})$.
We check whether $\text{checkRoute}(R,\vec{v}_i,\vec{v}_{i+1})=\texttt{true}$.
\end{itemize}
Finally, we check (using $\text{checkRoute}$) whether there is a route from $s$ to $v_1$ that respects the delay profile of $v_1$
and whether there is a route from $v_{f'}$ to $z$ for the delay profile of~$v_{f'}$. 

If, for any choices of $v_1, \dots, v_{f'}$ and $\vec{v_1}, \dots, \vec{v_{f'}}$,
all of the above checks succeed, then we have found a solution.
Otherwise we conclude that there is no solution.

\smallskip

Next, we show that we obtain the claimed running time bound. We analyze the running time of each of the steps of the algorithm.
\begin{enumerate}
\item Computing a minimum timed feedback vertex set takes $2^{\bigO(f)}\cdot(|V|+|E|)^{\bigO(1)}$ time~\cite{TemporalPathsUnderWaitingTime}.
\item There are $\bigO(4^f)$ partitions $\hat{X}_0\uplus\hat{X}_1\uplus\hat{X}_2\uplus\hat{X}_3=\hat{X}$ of $\hat{X}$.
\item There are $\bigO(f!)$ possible orderings for $\hat{X}_1\uplus\hat{X}_2\uplus\hat{X}_3$.
\item There are $f^{\bigO(x f)}$ delay profiles combinations for the vertices in $\hat{X}_1\uplus\hat{X}_2\uplus\hat{X}_3$ to consider.
\item For each complete guessing step, the check in the last step of the algorithm can be performed in polynomial time.
\end{enumerate}
It follows that we obtain the claimed running time bound.

Finally, we briefly sketch out why our algorithm is correct.
To this end, observe that if the checking step of the algorithm succeeds for a given guess, then $\GG$ contains an $x$-delay-robust $(s,z)$-route.
For the other direction, assume that $\GG$ contains an $x$-delay-robust $(s,z)$-route $R$.
We guess which vertices of $\hat{X}$ are visited by $R$ and in which order.
Furthermore, we guess the delay profiles for all vertices $v\in\hat{X}$ that are visited by $R$.
Now we know that all checks must succeed, otherwise we obtain a contradiction to the existence of~$R$.

\subsection{Underlying Feedback Edge Number}

In this section, we show that \drp{} admits an \CFPT{}-algorithm with respect to the feedback edge number of the underlying graph. 
Given a (static) undirected graph $G = (V,E)$, a feedback edge set $F \subseteq E$ is a set of edges, so that $G-F$ is acyclic.
The \emph{feedback edge number} is the cardinality of  a minimum feedback edge set of~$G$. Formally, we show the following.

\begin{theorem} \label{corr:drp_fpt_feedbackedge}
	\drp{} can be solved in $2^{\bigO(f)} \cdot (\abs{V} \cdot \abs{E} \cdot x^2) + \bigO(\abs{E} \cdot \log \abs{E})$ time,
	where $f$ is the feedback edge number of the underlying graph.
\end{theorem}

Casteigts et al.~\cite{TemporalPathsUnderWaitingTime} designed an \CFPT{}-algorithm for the so-called \textsc{Restless Temporal Path} parameterized by the feedback edge number of the underlying graph.
This algorithm can be applied to \drp{} as well with minor modifications.

The \CFPT{}-algorithm as given by Casteigts et al.~\cite{TemporalPathsUnderWaitingTime} consist of four steps (only the last step needs adaptation to our problem):
\begin{enumerate}
	\item Exhaustively remove vertices with degree $\le 1$ from $\under{\GG}$ (except $s$ and $z$).
	\item Compute a minimum feedback edge set $F$ of $\under{\GG}$. Let $f := \abs{F}$.
	\item Let $V^{\ge 3}$ denote all vertices of $\under{\GG}$ with degree at least three. 
		Partition the forest $\under{\GG} - F$ into a set of maximal paths $\mathcal{P}$ with endpoints in $F \cup V^{\ge 3} \cup \{s,z\}$, and intermediate vertices all of degree 2.
		It holds that $\abs{\mathcal{P}} \in \bigO(f)$.
	\item Any $(s,z)$-route in $G$ can be formed with feedback edges $F$ and paths from $\mathcal{P}$. 
		Enumerate all at most $2^{\bigO{(k)}}$ $(s,z)$-routes and check (using \cref{thm:polyforest}) if any of them forms a $x$-delay-robust $(s,z)$-route.
\end{enumerate}

Denote by $G = (V, E') = \under{\GG}$ the underlying graph.
Clearly it can be computed in $\bigO(\abs{\GG})$~time.
In step 1, vertices of degree $\le 1$ (except $s$ and $z$) can be removed safely, since they can never be part of an $(s,z)$-route.
This step is possible in $\bigO(|V|)$ time.
In step 2, the minimum feedback vertex set $F$ can be obtained by computing a spanning tree $S$ of $G$ in $\bigO(\abs{G})$ time, and taking $F = E' \setminus S$.
A key observation that is used in step 3 is that $\abs{V^{\ge 3}}$, the number of vertices of degree 3 and higher, in a graph without any vertex of degree 1 can be upper-bounded by $2\cdot f$ \cite{BentertDKNN20BetweennessCentrality}. 
Since $s$ and $z$ can have degree one and can not be removed from $\GG$ the actual upper bound is $2\cdot f + 2$.
(The deletion of a degree one vertex will lower the degree of exactly one other vertex in $G$,
and thus the deletion will only decrease the cardinality of $V^{\ge 3}$ by at most one.)
Thus, $\abs{F \cup V^{\ge 3} \cup \{s,z\}} \in \bigO(f)$ which are the endpoints of the paths $\mathcal{P}$.
Since $\under{\GG} - F$ is a forest we have $\abs{\mathcal{P}} \in \bigO(f)$.
The feedback edges and the paths $\mathcal{P}$ can be used to build all $(s,z)$-paths of~$G$.
Hence, the number of possible $(s,z)$-paths is bounded through the number of subsets of $F$ and $\mathcal{P}$,
which is itself in $2^{\bigO(f)}$.
We can now test for any $(s,z)$-path in $G$ whether it forms a robust route  by using the dynamic program from \cref{subsec:drp_verification}.
This requires an initial sorting of all time arcs in $\bigO(\abs{E} \log \abs{E})$~time,
plus $\bigO(\abs{V} \cdot \abs{E} \cdot x^2)$~time per path.
This leads to an overall time complexity of $2^{\bigO(f)} \cdot \abs{V} \cdot \abs{E} \cdot x^2 + \bigO(\abs{E} \log \abs{E})$.

\section{Conclusion}
We modeled a naturally motivated path-finding problem taking
into account delays
by means of (algorithmic) 
temporal graph theory. 
For our central problem, \drp{}, we found 
computational hardness already for some tree-like underlying (static) graphs.
While having provided a few encouraging parameterized tractability results, 
we leave plenty of room for further investigations into this direction.
In particular, we left open what happens for the special case when the number
of time labels per edge is bounded from above (in parameterized complexity 
terms, taking this as a parameter).
Recall that our central hardness reduction needs many time labels.
Moreover, the parameters vertex cover number or timed feedback vertex set 
number~\cite{TemporalPathsUnderWaitingTime} (as a single parameter) deserve investigations 
as well.
Rather from a modeling perspective, one might vary the basic problem 
by e.g.\ considering a \emph{global} delay budget or 
other variations of the 
delay concept.

\bibliography{strings-long,bibfile}

\end{document}